\documentclass[journal]{IEEEtran}

\usepackage{setspace}

\usepackage{graphics,
           psfrag,
           epsfig,
           amsthm,
           cite,
           amssymb,
           url,
           dsfont,
           subfigure,
           algorithm,
           algorithmic,
           balance,
           enumerate,
           color,
           setspace
}
\usepackage{amsmath}

\usepackage{epstopdf}

\newtheorem{lemma}{Lemma}

\newtheorem{theorem}{Theorem}
\newtheorem{corollary}{Corollary}
\newtheorem{proposition}{Proposition}

\newtheorem{remark}{Remark}

\newcommand{\eref}[1]{(\ref{#1})}
\newcommand{\sref}[1]{Section~\ref{#1}}
\newcommand{\appref}[1]{Appendix~\ref{#1}}
\newcommand{\fref}[1]{Figure~\ref{#1}}

\newcommand{\cref}[1]{Constraint~\ref{#1}}
\newcommand{\thref}[1]{Theorem~\ref{#1}}
\newcommand{\corref}[1]{Corollary~\ref{#1}}
\newcommand{\lref}[1]{Lemma~\ref{#1}}

\newcommand{\algref}[1]{Algorithm~\ref{#1}}

\newcommand{\ignore}[1]{}

\IEEEoverridecommandlockouts
\begin{document}

\title{\vspace{-.5cm}{Distributed Hybrid Scheduling in Multi-Cloud Networks using Conflict Graphs}}

\author{
   \authorblockN{Ahmed Douik, \textit{Student Member, IEEE}, Hayssam Dahrouj, \textit{Senior Member, IEEE},\\ Tareq Y. Al-Naffouri, \textit{Member, IEEE}, and Mohamed-Slim Alouini, \textit{Fellow, IEEE}\vspace{-.8cm}}

\thanks {A part of this paper \cite{3514852} is published in IEEE Global Telecommunications Conference (GLOBECOM' 2015), San Diego, CA, USA.

Ahmed Douik is with the Department of Electrical Engineering, California Institute of Technology, Pasadena, CA 91125 USA (e-mail: ahmed.douik@caltech.edu).

Hayssam Dahrouj is with the Department of Electrical Engineering, Effat University, Jeddah 22332, Saudi Arabia (e-mail: hayssam.dahrouj@gmail.com).

T. Y. Al-Naffouri and M.-S. Alouini are with the Division of Computer, Electrical and Mathematical Sciences, and Engineering, King Abdullah University of Science and Technology, Thuwal 23955-6900, Saudi Arabia (e-mail: \{tareq.alnaffouri,slim.alouini\}@kaust.edu.sa).
}
}

\maketitle

\begin{abstract}
Recent studies on cloud-radio access networks assume either signal-level or scheduling-level coordination. This paper considers a hybrid coordinated scheme as a means to benefit from both policies. Consider the downlink of a multi-cloud radio access network, where each cloud is connected to several base-stations (BSs) via high capacity links, and, therefore, allows for joint signal processing within the cloud transmission. Across the multiple clouds, however, only scheduling-level coordination is permitted, as low levels of backhaul communication are feasible. The frame structure of every BS is composed of various time/frequency blocks, called power-zones (PZs), which are maintained at a fixed power level. The paper addresses the problem of maximizing a network-wide utility by associating users to clouds and scheduling them to the PZs, under the practical constraints that each user is scheduled to a single cloud at most, but possibly to many BSs within the cloud, and can be served by one or more distinct PZs within the BSs' frame. The paper solves the problem using graph theory techniques by constructing the conflict graph. The considered scheduling problem is, then, shown to be equivalent to a maximum-weight independent set problem in the constructed graph, which can be solved using efficient techniques. The paper then proposes solving the problem using both optimal and heuristic algorithms that can be implemented in a distributed fashion across the network. The proposed distributed algorithms rely on the well-chosen structure of the constructed conflict graph utilized to solve the maximum-weight independent set problem. Simulation results suggest that the proposed optimal and heuristic hybrid scheduling strategies provide appreciable gain as compared to the scheduling-level coordinated networks, with a negligible degradation to signal-level coordination.
\end{abstract}


\section{Introduction}\label{sec:int}

Next generation mobile radio systems ($5$G) are expected to undergo major architectural changes, so as to support the deluge in demand for mobile data services by increasing capacity, energy efficiency and latency reduction \cite{6824752,6736746}. One way to boost throughput and coverage in dense data networks is by moving from the single high-powered base-station (BS) to the massive deployment of overlaying BSs of different sizes. Such architecture, however, is subject to high inter-BS interference, especially with the progressive move towards full spectrum reuse in $5$G. Traditionally, interference mitigation is performed by coordinating the different BSs through massive signaling and message exchange. Such coordination technique, however, in addition to being energy-inefficient \cite{6983623}, may not always be feasible given the capacity limits of the backhaul links.

A promising network architecture for fulfilling the ambitious metrics of 5G is the cloud-radio access network (CRAN) \cite{5594708,7143328}, which is obtained by connecting the different BSs to a central unit, known as the cloud. Such architecture moves most of the fundamental network functionalities to the cloud side, thereby allowing a separation between the control plane and the data plane. The virtualization in CRANs provides efficient resource utilization, joint BSs operation (joint transmission, encoding and decoding), and efficient energy control.

Different levels of coordination in CRANs are studied in the past literature, namely the signal-level coordination \cite{6799231,6588350,6786060}, and the scheduling-level coordination \cite{6525475,6811617,117665}. In signal-level coordinated CRANs \cite{6799231,6588350,6786060}, all the data streams of different users are shared among the different BSs, thereby allowing joint operation. However, such level of coordination necessitates high-capacity backhaul links. On the other hand, in scheduling-level coordinated CRANs \cite{6525475,6811617,117665}, the cloud is responsible only for the efficient allocation of the resource blocks of each BS, which requires much less backhauling. While more practical to implement, scheduling-level coordination may lead to an inferior performance as compared to signal-level coordination. {While clouds are typically connected to their base-stations through high-capacity links, cloud-to-cloud communication is done via wireless links. This paper, therefore, proposes a hybrid scheduling scheme which benefits from the advantages of both scheduling policies. In particular, the paper proposes using signal-level coordination within each cloud, and scheduling-level coordination among different clouds.}

Consider the downlink of a multi-CRAN, where each cloud is connected to several BSs. The frame structure of every BS is composed of various time/frequency blocks, called power-zones (PZs), kept at a fixed power level. This paper proposes a hybrid level of coordination for the scheduling problem. For BSs connected to the same cloud, associating users to PZs is performed assuming signal-level coordination. Across the multiple clouds, only scheduling-level coordination is permitted, as it requires a lower level of backhaul communication.

In this paper context, hybrid-level coordination refers to the scheme wherein multiple clouds coordinate their transmission on a scheduling-level basis only. Every cloud, however, is responsible for coordinating the transmission of its connected base-stations on a signal-level basis. The hybrid scheduling problem then denotes the strategy of assigning users to clouds across the network, under the system limitation that each user is scheduled at most to a single cloud since, otherwise, inter-cloud signal-level coordination is required. However, across the BSs connected to one cloud, users can be served by multiple BSs and different PZs within each transmit frame. Each PZ is further constrained to serve exactly one user.

\subsection{Related Work}

The paper is related in part to the classical works on scheduling, and in part to the recent works on CRAN. In the classical literature of cellular systems, scheduling is often performed assuming a prior assignment of users to BSs, e.g., the classical proportional fairness scheduling investigated in \cite{6525475,5464705}. In CRANs, recent works on coordinated scheduling consider a single cloud processing, as in \cite{6811617,117665}. Reference \cite{6811617} considers the particular case of coordinated scheduling when the number of users is equal to the number of available power-zones. Reference \cite{6811617} shows that, in a context of a soft-frequency reuse, the problem reduces to a classical linear programming problem that can be solved using the auction methodology \cite{49842514}. The problem is extended to an arbitrary number of users and power-zones in \cite{117665} and is shown to be an NP-hard problem. This paper is further related to the multi-cloud network studied in \cite{ouss_glob1,6799231} which, however, assume a pre-known user-to-cloud association.

Interference mitigation in CRANs via signal-level coordination has also been pivotal in the past few years. The authors in \cite{6588350} consider the problem of maximizing the weighted sum-rate under finite-capacity backhaul and transmit power constraints. Unlike previous studies in which compression is performed independently of the base-station operations, the authors in \cite{6588350} consider a joint precoding and backhaul compression strategy. Reference \cite{6786060} considers the problem of minimizing the total power consumption by accounting for the transport link power in a green-CRAN and proposes solving the problems using techniques from compressive sensing and optimization theory. Reference \cite{4450840} derives bounds on the achievable ergodic capacity to quantify the user diversity gain. In a classic multi-cell network setup, reference \cite{6831362} investigates the problem of joint beamforming design in a multi-cell system where multiple base-stations can serve each scheduled user. Using compressive sensing technique, reference \cite{6831362} illustrates the interplay between the transmit sum-power and the backhaul sum-capacity required to form the clusters, under fixed signal-to-interference-and-noise ratio (SINR) constraints.

All the aforementioned network optimization algorithms are centralized in nature, which is not always practically feasible for computational complexity reasons. This paper addresses this issue by proposing distributed algorithms so as to lessen the computational complexity and facilitate the practical implementation of the proposed methods. The paper is, therefore, related to the recent state-of-art on distributed scheduling, {e.g.,  \cite{4432271,5199027,5165179,4686832,6086561,6492306}.} While reference \cite{4432271} considers maximizing the capacity based on the complete co-channel gain information, reference \cite{5199027} considers the average channel state information only. Reference \cite{5165179}, on the other hand, proposes a distributed algorithm for interference mitigation which automatically adjusts the transmit power in orthogonal frequency division multiple access (OFDMA) based cellular systems. Reference \cite{4686832}, further, investigates a distributed scheduling approach to maximize the sum-rate using zero-forcing beamforming in a multiple-antenna base-station setup. Reference \cite{4686832}, particularly, shows that distributing the computations of the scheduling problem across the users provides satisfactory results by decreasing the computational complexity and reducing overhead. A survey on useful distributed techniques can be found in \cite{6492306}. 

The distributed algorithms presented in this paper are also related to solutions suggested in \cite{6086561,4389757,1626432,4558622,4151582}. References \cite{6086561,4389757,1626432,4558622,4151582}, however, are based on a game-theoretical formulation of the resource allocation problem. For instance, reference \cite{4389757} proposes a distributed algorithm for resource allocation and adaptive transmission in a multi-cell scenario, which enables a trade-off between the aggressive reuse of the spectrum and the consequent co-channel interference. Similarly, in \cite{4151582}, the balance between the power and resource allocation is investigated. Finally, while reference \cite{1626432} considers the distributed power control scheme in wireless ad hoc networks, reference \cite{4558622} studies the problem in wireless OFDM systems.

\subsection{Contributions}

Unlike the aforementioned references, this papers considers the downlink of a multi-CRAN, where each cloud is connected to several base-stations (BSs) via high capacity links and, therefore, allows for joint signal processing within the cloud transmission. Across the multiple clouds, however, only scheduling-level coordination is permitted. The frame structure of every BS is composed of various power-zones, which are maintained at a fixed power level. The paper then addresses the coordinated scheduling with an objective of maximizing a generic utility function. The paper's main contribution is to solve the problem optimally using techniques inherited from graph theory. The paper proposes both optimal and heuristic distributed solutions to the problem. The paper also explicitly characterizes the extremes in scheduling policies, i.e., either scheduling-level or signal-level coordination, and proposes solving the problems using graph-theory based algorithms.

The first part of the paper investigates the centralized coordinated scheduling problem. It considers the architecture wherein all the clouds are connected to a central processor that is responsible for computing the scheduling policy and maintaining the synchronization of the different transmit frames. The paper proposes solving the hybrid scheduling problem by constructing the conflict graph, in which each vertex represents an association of cloud, user, base-station and power-zone. The solution then relies on reformulating the problem as a maximum-weight independent set problem that can be optimally solved using efficient algorithms, e.g., \cite{15522856,2155446,6848102,5341909}.

The second part of the paper investigates the distributed coordinated scheduling problem. It considers the scenario wherein the different clouds are connected through low capacity links. In this configuration, the optimal scheduling decision is reached through intelligent, reasonable information exchange among the clouds. The distributed solution is achieved via the construction of local conflict graphs and the local solutions of the maximum-weight independent set problem. In order to produce a feasible solution, a conflict resolution phase comes afterward to ensure that each user is scheduled to at most a single cloud. Further, the paper proposes a low complexity, heuristic, distributed solution that relies on strictly assigning users to clouds according to the highest utility.

Finally, the paper considers both the scheduling-level and signal-level coordination separately and shows how each setup can be solved as a particular case of the generic framework. The paper simulation results suggest that the proposed hybrid scheduling strategy provides appreciable gain as compared to the scheduling-level coordinated networks, with a negligible degradation to signal-level coordination.

The rest of this paper is organized as follows: In \sref{sec:sys}, the system model, and the problem formulation are presented. \sref{sec:mul} proposes a solution to the hybrid scheduling problem. In \sref{sec:dis}, optimal and heuristic distributed solutions are presented. \sref{sec:ful} presents the scheduling solution of signal and scheduling level coordinated networks. Simulation results are discussed in \sref{sec:sim}, and conclusions are presented in \sref{sec:con}.

\section{System Model and Problem Formulation}\label{sec:sys}

\subsection{System Model and Parameters}

Consider the downlink of a multi-CRAN of $C$ clouds serving $U$ users in total. The $C$ clouds are connected to a central cloud. Each cloud (except the central one) is connected to $B$ BSs and is responsible for the signal-level coordination of the connected BSs. \fref{fig:network1} illustrates a multi-CRAN formed by $U=21$ users, and $C=3$ clouds each coordinating $B=3$ BSs. Let $\mathcal{C}$ be the set of clouds in the system each coordinating the set of BSs $\mathcal{B}$. All BSs and users are equipped with single antennas. Let $\mathcal{U}$ be the set of users in the network ($|\mathcal{U}|= U$, where the notation $|\mathcal{X}|$ refers to the cardinality of a set $\mathcal{X}$). The transmit frame of each BS is composed of several time/frequency resource blocks maintained at fixed transmit power. In this paper, the generic term PZ is used to refer to a time/frequency resource block of a BS. Let $\mathcal{Z}$ be the set of the $Z$ PZs of the frame of one BS. The transmit power of the $z$th PZ in the $b$th BS of the $c$th cloud is fixed to $P_{cbz}$, $\forall \ (c,b,z) \in \mathcal{C} \times \mathcal{B} \times \mathcal{Z}$, where the notation $\mathcal{X} \times \mathcal{Y}$ refers to the Cartesian product of the two sets $\mathcal{X}$ and $\mathcal{Y}$. \fref{fig:frame} shows the coordinated frames of the connected BSs in the $c$th cloud. This paper focuses on the scheduling optimization (i.e., for a fixed transmit paper. Optimization with respect to the power values $P_{cbz}$ is left for future research.)

\begin{figure}[t]
\centering
\includegraphics[width=0.7\linewidth]{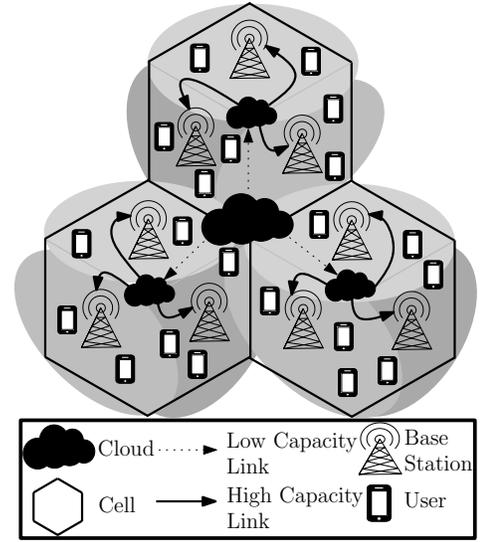}
\caption{Cloud enabled network composed $3$ cells, each containing $3$ base stations and $7$ users.}\label{fig:network1}
\end{figure}
\begin{figure}[t]
\centering
\includegraphics[width=1\linewidth]{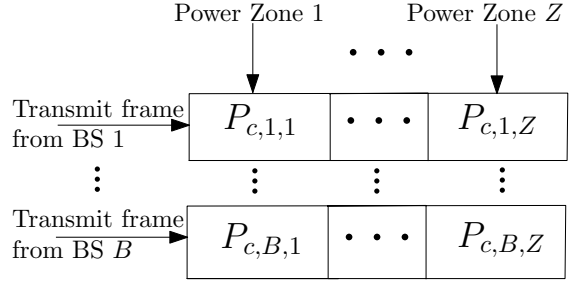}
\caption{Frame structure of $B$ base stations each containing $Z$ power zones.}\label{fig:frame}
\end{figure}

Each cloud $c \in \mathcal{C}$ is responsible for coordinating its $B$ BSs, which allows joint signal processing across them. The central cloud connecting all the clouds $c \in \mathcal{C}$ is responsible for computing the scheduling policy, and also guarantees that the transmission of the different frames are synchronized across all BSs in the network ($CB$ BSs). Let $h_{cbz}^{u} \in \mathds{C},\ \forall \ (c,u,b,z) \in \mathcal{C} \times \mathcal{U} \times \mathcal{B} \times \mathcal{Z}$ be the complex channel gain from the $b$th BS of the $c$th cloud to user $u$ scheduled to PZ $z$. The signal-to-interference plus noise-ratio (SINR) of user $u$ when scheduled to PZ $z$ in the $b$th BS of the $c$th cloud can be expressed as:
\begin{align}
\text{SINR}_{cbz}^{u} = \cfrac{P_{cbz} |h_{cbz}^{u}|^2}{\Gamma(\sigma^2+ \sum\limits_{(c^\prime,b^\prime) \neq (c,b)} P_{c^\prime b^\prime z}|h_{c^\prime b^\prime z}^{u}|^2)},
\end{align}
where $\Gamma$ denotes the SINR gap, and $\sigma^2$ is the Gaussian noise variance. This paper assumes that the cloud is able to perfectly estimate all the values of the channel gains $h_{cbz}^{u}$ and thus the different SINRs. 

\subsection{Scheduling Problem Formulation}

The scheduling problem under investigation in this paper consists of assigning users to clouds and scheduling them to PZs in each BS frame under the following practical constraints.
\begin{itemize}
\item C1: Each user can connect at most to one cloud but possibly to many BSs in that cloud.
\item C2: Each PZ should be allocated to exactly one user.
\item C3: Each user cannot be served by the same PZ across different BSs.
\end{itemize}

Let $\pi_{cubz}$ be a generic network-wide benefit of assigning user $u$ to the $z$th PZ of the $b$th BS in the $c$th cloud. Let $X_{cubz}$ be a binary variable that is $1$ if user $u$ is mapped to the $z$th PZ of the $b$th BS in the $c$th cloud, and zero otherwise. Similarly, let $Y_{uz}$ be a binary variable that is $1$ if user $u$ is mapped to the $z$th PZ of any BS across the network, and zero otherwise. Further, let $Z_{cu}$ be a binary variable that is $1$ if user $u$ is assigned to cloud $c$. The scheduling problem this paper addresses can be formulated as the following 0-1 mixed integer programming problem:
\begin{subequations}
\label{Original_optimization_problem}
\begin{align}
\max & \sum_{c,u,b,z}\pi_{cubz}X_{cubz} \\
\label{constraint1}{\rm s.t.\ } &Z_{cu} = 1 - \delta \bigg(\sum_{b,z}X_{cubz}\bigg), \forall \ (c,u)\in\mathcal{C} \times \mathcal{U},\\
\label{constraint2}& \sum_{c}Z_{cu} \leq 1 , \quad \forall \ u \in\mathcal{U},\\
\label{constraint3}& \sum_{u}X_{cubz}=1, \quad\forall \ (c,b,z)\in \mathcal{C} \times \mathcal{B} \times \mathcal{Z},\\
\label{constraint4}& Y_{uz} = \sum_{cb} X_{cubz} \leq 1,\quad \forall \ (u,z) \in \mathcal{U} \times \mathcal{Z}, \\
\label{constraint5}& X_{cubz},Y_{uz},Z_{cu} \in \{0,1\},
\end{align}
\end{subequations}
where the optimization is over the binary variables $X_{cubz}$, $Y_{uz}$, and $Z_{cu}$ and the notation $\delta(.)$ refers to the discrete Dirac function which is equal to $1$ if its argument is equal to $0$ and $0$ otherwise. Both the equality constraint \eref{constraint1} and the inequality constraint \eref{constraint2} are due to system constraint C1. The equality constraints \eref{constraint3} and \eref{constraint4} correspond to the system constraints C2 and C3, respectively.

Using a generic solver for 0-1 mixed integer programs may require a search over the entire feasible space of solutions, i.e., all possible assignments of users to clouds and PZs of the network BSs. The complexity of such method is prohibitive for any reasonably sized system. The next section, instead, presents a more efficient method to solve the problem by constructing the conflict graph in which each vertex represents an association between clouds, users, BSs, and PZs. The paper reformulates the 0-1 mixed integer programming problem \eref{Original_optimization_problem} as a maximum-weight independent set problem in the conflict graph, which global optimum can be reached using efficient techniques, e.g., \cite{15522856,2155446}.

\section{Multi-Cloud Coordinated Scheduling}\label{sec:mul}

This section presents the optimal solution to the optimization problem \eref{Original_optimization_problem} by introducing the conflict graph and reformulating the problem as a maximum-weight independent set problem. The corresponding solution is naturally centralized, and the computation must be carried at the central cloud connecting all the clouds $c \in \mathcal{C}$.

\subsection{Conflict Graph Construction}

Define $\mathcal{A} = \mathcal{C} \times \mathcal{U} \times \mathcal{B} \times \mathcal{Z}$ as the set of all associations between clouds, users, BSs, and PZs, i.e., each element $a \in \mathcal{A}$ represents the association of one user to a cloud and a PZ in one of the connected BSs frame. For each association $a=(c,u,b,z) \in \mathcal{A}$, let $\pi(a)$ be the benefit of such association defined as $\pi(a)=\pi_{cubz}$. Let $\varphi_c$ be the cloud association function that maps each element from the set $\mathcal{A}$ to the corresponding cloud in the set $\mathcal{C}$. In other words, for $a=(c,u,b,z) \in \mathcal{A}$, $\varphi_c(a)=c$. Likewise, let $\varphi_u$, $\varphi_b$, and $\varphi_z$ be the association functions mapping each element $a=(c,u,b,z) \in \mathcal{A}$ to the set of users $\mathcal{U}$ (i.e., $\varphi_u(a)=u$), to the set of BSs $\mathcal{B}$ (i.e., $\varphi_b(a)=b$), and to the set of PZs (i.e., $\varphi_z(a)=z$), respectively.

The \emph{conflict} graph $\mathcal{G}(\mathcal{V},\mathcal{E})$ is an undirected graph in which each vertex represents an association of cloud, user, BS and PZ. Each edge between vertices represents a conflict between the two corresponding associations. Therefore, the conflict graph can be constructed by generating a vertex $v \in \mathcal{V}$ for each association $a \in \mathcal{A}$. Vertices $v$ and $v^\prime$ are conflicting vertices, and thus connected by an edge in $\mathcal{E}$ if one of the following connectivity conditions (CC) is true:
\begin{itemize}
\item CC1: $\delta(\varphi_u(v) - \varphi_u(v^\prime))(1-\delta(\varphi_c(v) - \varphi_c(v^\prime))) = 1$.
\item CC2: $(\varphi_c(v),\varphi_b(v),\varphi_z(v)) = (\varphi_c(v^\prime),\varphi_b(v^\prime),\varphi_z(v^\prime))$.
\item CC3: $\delta(\varphi_u(v) - \varphi_u(v^\prime))\delta(\varphi_z(v) - \varphi_z(v^\prime)) = 1$.
\end{itemize}

The connectivity constraint CC1 corresponds to a violation of the system constraint C1 as it describes that two vertices are conflicting if the same user is scheduled to different clouds. The connectivity constraint CC2 partially illustrates the system constraint C2, as it implies that each PZ should be associated with at most one user (not exactly one user as stated in the original system constraint). With the additional constraint (see \thref{th1} below) about the size of the independent set, CC2 becomes equivalent to C2. Finally, the edge creation condition CC3 correctly translates a violation of the system constraint C3.

\begin{figure}[t]
\centering
\includegraphics[width=0.8\linewidth]{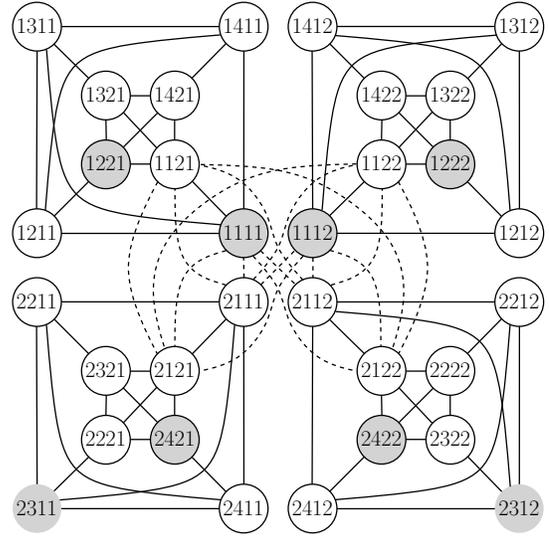}\\
\caption{Example of the conflict graph for a network composed of $2$ clouds, $2$ BSs per cloud, $2$ PZs per BS and a total of $4$ users. Intra-cloud connection are plotted in solid lines. Inter-cloud connections are illustrated only for user $1$ in dashed lines.}
\label{fig:graph}
\end{figure}

\fref{fig:graph} illustrates an example of the conflict graph in a multi-cloud system composed of $C=2$ clouds, $B=2$ BSs per cloud, $Z=2$ PZs per BS and $U=4$ users. Vertices, in this example, are labelled $cubz$, where $c$, $u$, $b$ and $z$ represent the indices of clouds, users, BSs, and PZs, respectively. In this example, $Z_{\text{tot}}=CBZ=8$. As shown in \fref{fig:graph}, each independent set of size $Z_{\text{tot}}$ can be written in the following form:
\begin{enumerate}
\item \{$1a11$,\ $1a12$,\ $1b21$,\ $1b22$,\ $2c11$,\ $2c12$,\ $2d21$,\ $2d22$\}
\item \{$1a11$,\ $1a12$,\ $1b21$,\ $1b22$,\ $2c11$,\ $2d12$,\ $2d21$,\ $2c22$\}
\item \{$1a11$,\ $1b12$,\ $1b21$,\ $1a22$,\ $2c11$,\ $2c12$,\ $2d21$,\ $2d22$\}
\item \{$1a11$,\ $1b12$,\ $1b21$,\ $1a22$,\ $2c11$,\ $2d12$,\ $2d21$,\ $2c22$\},
\end{enumerate}
where $a,b,c,$ and $d$ $\in \{1,2,3,4\}$ with $a \neq b \neq c \neq d$. For example, replacing $(a,b,c,d)$ in \{$1a11$,\ $1a12$,\ $1b21$,\ $1b22$,\ $2c11$,\ $2c12$,\ $2d21$,\ $2d22$\} by $(1,2,3,4)$ gives the independent set shown in gray in \fref{fig:graph}, which is a set of non-connected vertices of size $Z_{\text{tot}}=8$. The $4!=24$ distinct permutations of $(a,b,c,d)$ eventually result in $4!\times 4=96$ independent sets of size $Z_{\text{tot}}$ in total.

\subsection{Scheduling Solution}

Consider the conflict graph $\mathcal{G}(\mathcal{V},\mathcal{E})$ constructed above and let $\mathcal{I}$ be the set of all independent set of vertices of size $Z_{\text{tot}} = CBZ$. The following theorem characterises the solution of the optimization problem \eref{Original_optimization_problem}.

\begin{theorem}
The global optimal solution to the scheduling problem in multi-cloud network \eref{Original_optimization_problem} is the maximum-weight independent set among the independent sets of size $Z_{\text{tot}}$ in the conflict graph, where the weight of each vertex $v \in \mathcal{V}$ is given by:
\begin{align}
w(v) = \pi(v).
\label{eeeq2}
\end{align}
In other words, the optimal solution of the scheduling problem \eref{Original_optimization_problem} can be expressed as:
\begin{align}
I^* = \arg \max_{I \in \mathcal{I}} \sum_{v \in I} w(v).
\end{align}
\label{th1}
\end{theorem}
\begin{proof}
A sketch of the proof goes as follows. The optimization problem \eref{Original_optimization_problem} is first reformulated as a search over the set of feasible schedules. Further, a one to one mapping between the possible schedules and the set of independent sets of size $Z_{tot}$ in the conflict graph is established. Finally, showing that the weight of each independent set is the objective function of \eref{Original_optimization_problem} indicates that the optimal solution is the maximum-weight independent set, which concludes the proof. A complete proof can be found in \appref{app1}.
\end{proof}

\subsection{Complexity Analysis and Heuristic Algorithm}

In graph theory context, an independent set is a set in which each two vertices are not adjacent. The maximum-weight independent set problem is the problem of finding, in a weighted graph, the independent set(s) with the maximum weight where the weight of the set is defined as the sum of the individual weights of vertices belonging to the set. Maximum-weight independent set problems are well-known NP-hard problems. However, they can be solved efficiently, e.g., \cite{15522856,2155446}. Therefore, the complexity of the proposed solution can be written as $\mathds{C}_{\text{cen}}^{\text{opt}} = \alpha^{CBZU}$, where $1 < \alpha \leq 2$ is a constant that depends on the applied algorithm, e.g., $\alpha=1.21$ for \cite{2155446}. Moreover, several approximate \cite{6848102} and polynomial time \cite{5341909} methods produce satisfactory results, in general. {This subsection presents a heuristic, yet simple, algorithm which discovers a maximal\footnote{{A maximal independent set is a set that is no longer independent if any node is added to it. The maximum independent set is the maximum of all such maximal sets.}} weight independent set.}

\begin{algorithm}[t]
\begin{algorithmic}
\REQUIRE $\mathcal{C}$, $\mathcal{U}$, $\mathcal{B}$, $\mathcal{Z}$, $P_{bz}$, and $h_{cbz}^{u},\ \forall \ c \in \mathcal{C}, u \in \mathcal{U},\ \forall \ b \in \mathcal{B},\ \forall\ z \in \mathcal{Z}$
\STATE Initialize $\mathbf{S} = \varnothing$.
\STATE Construct $\mathcal{G}$ using subsection III-A.
\STATE Compute weight $w(v), \ \forall \ v \in \mathcal{G}$ using \eref{eeeq2}.
\WHILE{$\mathcal{G} \neq \varnothing$}
\STATE Select $v^* = \text{argmax }_{v \in \mathcal{G}} w(v)$.
\STATE Set $\mathbf{S} = \mathbf{S} \cup \{v^*\}$
\STATE Set $\mathcal{G}= \mathcal{G}(v^*)$ where $ \mathcal{G}(v^*)$ is the sub-graph of $\mathcal{G}$ containing only the vertices not adjacent to $v^*$.
\ENDWHILE
\STATE Output $\mathbf{S}$.
\end{algorithmic}
\caption{Independent set search heuristic.}
\label{alg1}
\end{algorithm}

{To solve the maximum-weight independent set problem in linear time with the size of the graph, a simple procedure is to sequentially select nodes with largest weights. First, construct the graph $\mathcal{G}$. The idea here is to sequentially update the independent set $\mathbf{S}$ by adding the vertex with the highest weight at each step. Then, the graph is updated by removing all vertices adjacent to the selected vertex, so as to guarantee that the connectivity constraints CC1, CC2, and C3 are satisfied. The process is repeated until the graph becomes empty. The steps of the heuristic are summarized in \algref{alg1}}.

\section{Distributed Coordinated Scheduling} \label{sec:dis}

\begin{figure}[t]
\centering
\includegraphics[width=0.7\linewidth]{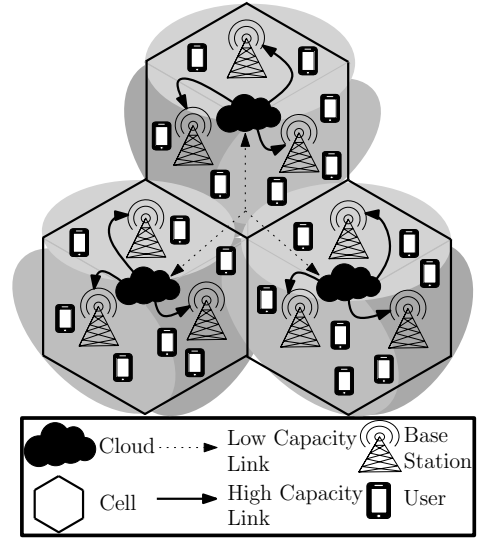}
\caption{Cloud enabled network composed $3$ cells, each containing $3$ base stations and $7$ users.}\label{fig:network2}
\end{figure}

The previous section assumes the presence of a central cloud that is responsible for computing the scheduling policy in a centralized fashion, which may not always be feasible from a computational complexity perspective. This section instead considers the multi-CRAN in which clouds are interconnected through low capacity links, as shown in \fref{fig:network2}. In contrast to the centralized system in \fref{fig:network1}, the joint scheduling is now performed under the constraint that each cloud $c \in \mathcal{C}$ has partial access to the network parameters. In particular, cloud $c \in \mathcal{C}$ has knowledge of the its channel gains only, i.e., $h_{cbz}^{u},\ \forall \ (u,b,z) \in \mathcal{U} \times \mathcal{B} \times \mathcal{Z}$. The distributed joint scheduling problem becomes the one of scheduling users to clouds and PZs in the connected BSs, by only allowing a reasonable amount of information exchange among the clouds.

\begin{remark}
Note that exchanging all the network parameters, i.e., $h_{cbz}^{u},\ \forall \ (c,u,b,z) \in \mathcal{C} \times \mathcal{U} \times \mathcal{B} \times \mathcal{Z}$, constructing the conflict graph at each cloud and solving the maximum-weight independent set may solve the optimization problem \eref{Original_optimization_problem}. However, such solution not only requires a considerable amount of backhaul communication (exchange of $CUBZ$ complex variable), but also a waste of computation resources as the $C$ clouds solve the same problem.
\end{remark}

The first part of this section provides the optimal distributed coordinated scheduling. In other words, this part characterizes the solution of the optimization problem \eref{Original_optimization_problem} by solving the maximum-weight independent set in the conflict graph in a distributed fashion. The second part of this section provides a heuristic, low complexity, distributed solution for problem \eref{Original_optimization_problem}.

\subsection{Optimal Distributed Coordinated Scheduling}

To solve the scheduling problem \eref{Original_optimization_problem} in a distributed fashion, the paper proposes a distributed method to resolve the maximum-weight independent set using the particular structure of the conflict graph. Before describing the steps of the algorithm, the section first introduces the local scheduling graph $\mathcal{G}_c (\mathcal{U}_c)$ for an arbitrary cloud $c \in \mathcal{C}$ and its set of scheduled users $c$, called $\mathcal{U}_c \subseteq \mathcal{U}$.

\begin{algorithm}[t]
\begin{algorithmic}
\REQUIRE $\mathcal{C}$, $\mathcal{U}$, $\mathcal{B}$, $\mathcal{Z}$, $P_{cbz}$, and $h_{cbz}^{u}$.
\STATE \% Initial Phase
\STATE Initialize $\mathcal{K} = \varnothing$.
\FORALL {$c \in \mathcal{C}$}
\STATE Initialize $\mathcal{U}_c = \mathcal{U}$.
\STATE Construct local conflict graph $\mathcal{G}_c (\mathcal{U}_c)$ and weights $w(v)$.
\STATE Solve $\mathbf{S}_c$ max-weight independent set of size $BZ$.
\ENDFOR
\STATE \% Conflict Resolution Phase
\FOR {$t=1,\ 2,\ \cdots$}
\STATE Broadcast $\hat{\mathbf{S}}_c = \left\{u \in \mathcal{U}_c \ | \ u \in \mathbf{S}_c\right\}$.
\STATE Set $\mathcal{K} = \left\{ u \in \mathcal{U} \ | \ \exists \ (c,c^\prime) \in \mathcal{C}^2, u \in \hat{\mathbf{S}}_c \cap \hat{\mathbf{S}}_{c^\prime}\right\}$.
\FORALL {$u \in \mathcal{K}$}
\STATE Set $\hat{\mathcal{C}}(u)= \left\{ c \in \mathcal{C} \ | \ u \in \hat{\mathbf{S}}_c \right\}$.
\FORALL {$c \in \hat{\mathcal{C}}(u)$}
\STATE Set $\pi_{cu} = \sum_{u^\prime,b,z} X_{cu^\prime bz} \pi_{cu^\prime bz}$.
\STATE Set $\mathcal{U}_c = \mathcal{U}_c \setminus \{u\}$.
\STATE Construct $\mathcal{G}_c (\mathcal{U}_c)$ and compute weights $w(v)$.
\STATE Solve $\overline{\mathbf{S}}_c$ max-weight independent set of size $BZ$.
\STATE Set $\overline{\pi}_{cu} = \sum_{u^\prime,b,z} X_{cu^\prime bz} \pi_{cu^\prime bz}$.
\STATE Broadcast $\pi_{c}$ and $\overline{\pi}_{c}$.
\ENDFOR
\STATE Set $c^* = \arg \max_{c \in \hat{\mathcal{C}}(u)} \left( \pi_{c} + \sum_{\substack{ c^\prime \in \hat{\mathcal{C}}(u) \\ c^\prime \neq c}} \overline{\pi}_{c^\prime } \right)$.
\STATE Set $\mathcal{U}_{c^*} = \mathcal{U}_{c^*} \cup \{u\}$.
\FORALL {$c \in \hat{\mathcal{C}}(u) \setminus \{c^*\}$}
\STATE Set $\mathbf{S}_c = \overline{\mathbf{S}}_c$
\ENDFOR
\ENDFOR
\ENDFOR
\STATE Output final schedule $\mathbf{S} = \bigcup_{c \in \mathcal{C}} \mathbf{S}_c$.
\end{algorithmic}
\caption{Distributed Coordinated Scheduling Algorithm}
\label{algo1}
\end{algorithm}

Let the reduced set of association of cloud $c$ be defined as $\hat{\mathcal{A}} = c \times \mathcal{U}_c \times \mathcal{B} \times \mathcal{Z}$. This set represents all associations cloud $c$ can perform when it is allowed to schedule users in the set $\mathcal{U}_c$. Note that all the benefits of the associations $\hat{a} \in \hat{\mathcal{A}}$ can be computed locally at cloud $c$ since all the needed complex channel gains $h_{cbz}^{u},\ \forall \ (u,b,z) \in \mathcal{U} \times \mathcal{B} \times \mathcal{Z}$ and power levels $P_{cbz}$, $\forall \ (b,z) \in \mathcal{B} \times \mathcal{Z}$ are locally available. The local conflict graph $\mathcal{G}_c (\mathcal{U}_c)$ is constructed in a similar manner as the conflict graph $\mathcal{G}$ except that it only considers associations $\hat{a} \in \hat{\mathcal{A}}$ in the vertex generation step. Hence, instead of containing $CUBZ$ vertex, the local conflict graph contains $U_cBZ$ where $U_c=|\mathcal{U}_c|$. The vertex connectivity conditions are the same as for the conflict graph.

The algorithm is composed of two phases, namely, the initialization and conflict resolution phases. In an initial phase, each cloud generates its local conflict graph and solves the maximum-weight independent set of size $BZ$. Each cloud communicates its scheduled users with the remaining clouds. It is worth mentioning that only the scheduled users are shared and not the complete information about the schedule (i.e., the PZs and the BSs in which they are scheduled).

After the initial phase, a conflict resolving step takes place. In this stage, users that are scheduled to multiple clouds are assigned to the cloud that generates the highest sum-benefit of scheduling that user across its multiple BSs and PZs. Clouds $c'$ that fail to have the maximum benefit are not allowed to schedule that user in the subsequent phases of the algorithm. Clouds $c'$, therefore, remove that user from their set of authorized users. The new graph is then constructed, and the maximum weight clique is subsequently solved. This process is repeated until all users are assigned to at most one single cloud. The steps of the algorithm are summarized in \algref{algo1}.

The following theorem characterizes the distributed solution reached by \algref{algo1}:
\begin{theorem}
\algref{algo1} converges to the optimal solution of the centralized coordinated scheduling optimization problem \eref{Original_optimization_problem} in at most $C(U-B)$ iterations.
\label{th2}
\end{theorem}

\begin{proof}
To show that the distributed solution reached by \algref{algo1} is the optimal solution to the scheduling problem \eref{Original_optimization_problem}, we first show that solving the maximum-weight independent set locally yields the optimal solution to \eref{Original_optimization_problem}, whenever the local and the global approaches start with users assigned to the same cloud. Afterward, we show that \algref{algo1} assigns users to clouds that coincide with the cloud assignment found through the optimal solution. To finish the proof, we show that the running time of the algorithm is bounded. A complete proof of the theorem can be found in \appref{app2}.
\end{proof}

\subsection{Heuristic Distributed Coordinated Scheduling}

In this section, a heuristic, low-complexity, distributed solution is presented. The algorithm follows the same steps as \algref{algo1}, except in the way of updating the local maximum-weight independent set at each cloud. While \algref{algo1} recomputes the new graph and the maximum-weight independent set for each user in conflict, the proposed low-complexity distributed solution updates the solution obtained in the previous round. In other words, instead of generating the new graph and recomputing the maximum-weight independent set at each step, the heuristic algorithm updates the maximum-weight clique obtained previously by removing vertices in conflict and adding new vertices, which simplifies the computational complexity.

\begin{algorithm}[t]
\begin{algorithmic}
\REQUIRE $\mathcal{C}$, $\mathcal{U}$, $\mathcal{B}$, $\mathcal{Z}$, $P_{cbz}$, and $h_{cbz}^{u}$.
\STATE \% Initial Phase as in \algref{algo1}
\STATE \% Conflict Resolution Phase
\FOR {$t=1,\ 2,\ \cdots$}
\STATE Broadcast $\hat{\mathbf{S}}_c = \left\{u \in \mathcal{U}_c \ | \ u \in \mathbf{S}_c\right\}$.
\STATE Set $\mathcal{K} = \left\{ u \in \mathcal{U} \ | \ \exists \ (c,c^\prime) \in \mathcal{C}^2, u \in \hat{\mathbf{S}}_c \cap \hat{\mathbf{S}}_{c^\prime}\right\}$.
\FORALL {$u \in \mathcal{K}$}
\STATE Set $\hat{\mathcal{C}}(u)= \left\{ c \in \mathcal{C} \ | \ u \in \hat{\mathbf{S}}_c \right\}$.
\FORALL {$c \in \hat{\mathcal{C}}(u)$}
\STATE Broadcast $\pi_{cu} = \sum_{b,z} X_{cubz} \pi_{cubz}$.
\ENDFOR
\STATE Set $c^* = \arg \max_{c \in \hat{\mathcal{C}}(u)} \pi_{cu}$.
\FORALL {$c \in \mathcal{C} \setminus \{c^*\}$}
\STATE Set $\mathcal{U}_c = \mathcal{U}_c \setminus \{u\}$.
\STATE Compute $\mathbf{S}_c^u = \left\{ a \in \mathbf{S}_c \ | \ \varphi_u(a)=u \right\}$.
\STATE Set $\mathbf{S}_c = \mathbf{S}_c \setminus \mathbf{S}_c^u$ and $\mathcal{V}_c=\mathcal{V}_c(\mathbf{S}_c)$.
\STATE Compute weights $w(v), \ \forall \ v \in \mathcal{V}_c$.
\STATE Solve $\tilde{\mathbf{S}}_c$ max-weight independent set of size $|\mathbf{S}_c^u|$.
\STATE Set $\mathbf{S}_c = \mathbf{S}_c \cup \tilde{\mathbf{S}}_c$.
\ENDFOR
\ENDFOR
\ENDFOR
\STATE Output final schedule $\mathbf{S} = \bigcup_{c \in \mathcal{C}} \mathbf{S}_c$.
\end{algorithmic}
\caption{Low Complexity Distributed Scheduling Scheme}
\label{algo2}
\end{algorithm}

To explicitly define the maximum-weight independent set update strategy, first, define $\mathbf{S}_c$ as the schedule obtained by solving the maximum-weight independent set, and $\mathbf{S}_c^u \subset \mathbf{S}_c$ as the set of vertices of user $u$ scheduled in cloud $c$. Further, let $\mathcal{V}_c(\mathbf{S}_c)$ be the set of vertices in the local conflict graph $\mathcal{G}_c(\mathcal{U}_c)$ that are not connected to any vertex in $\mathbf{S}_c$. Note that the vertices in $\mathcal{V}_c(\mathbf{S}_c)$ are combinable with the previous schedule $\mathbf{S}_c$ since they are not connected to any vertex in the schedule.

The low-complexity distributed algorithm follows the same steps in the initial phase as \algref{algo1}. In the conflict resolution phase, users that are scheduled to multiple clouds are assigned to the one with the highest sum-benefit. The remaining clouds remove the associations containing the user, i.e., vertices in $\mathbf{S}_c^u$, from their schedule. Afterward, they update their local conflict graph to only keep the vertices $\mathcal{V}_c(\mathbf{S}_c)$ that are not connected to all vertices previously selected in the schedule $\mathbf{S}_c \setminus \mathbf{S}_c^u$. The maximum-weight independent set of size $|\mathbf{S}_c^u|$ is then computed and appended to $\mathbf{S}_c$ to produce the schedule. The process is repeated until all users are assigned to at most one single cloud. The steps of the algorithm are summarized in \algref{algo2}.

\begin{corollary}
\algref{algo2} converges to a feasible solution of the centralized coordinated scheduling optimization problem \eref{Original_optimization_problem} in at most $U$ iterations.
\label{cor1}
\end{corollary}

\begin{proof}
To prove this corollary, it is sufficient to show that \algref{algo2} converges. Showing that the outputted schedule is a feasible one concludes the proof. A complete proof can be found in \appref{app3}.
\end{proof}

\subsection{Complexity Analysis}

This subsection compares the complexity of the optimal and heuristic distributed algorithms against the optimal centralized solution proposed in \sref{sec:mul}.

As shown in \algref{algo1}, each cloud solves a maximum weight independent set at each step of the algorithm in which it has conflicts with other clouds. The size of the scheduling graph of the $c$-th cloud is $CBZ|\mathcal{U}_c|$. Given the result in \thref{th2}, each cloud experiences $U-B$ conflict in the worst case. Therefore, the total complexity $\mathds{C}_{\text{dis}}^{\text{opt}}$ per cloud can be written as:
\begin{align}
\mathds{C}_{\text{dis}}^{\text{opt}} &= \sum_{|\mathcal{U}_c| = U-B}^U \alpha^{CBZ|\mathcal{U}_c|} = \alpha^{CBZU}\cfrac{\alpha^{-CB^2Z} - \alpha^{CBZ}}{1-\alpha^{CBZ}} \nonumber \\
&= \mathds{C}_{\text{cen}}^{\text{opt}} \cfrac{\alpha^{-CB^2Z} - \alpha^{CBZ}}{1-\alpha^{CBZ}} \label{eq:com}
\end{align}
where $1 < \alpha \leq 2$ is a constant that depends on the algorithm utilized in solving the maximum weight independent set problem. It can readily be seen from \eref{eq:com} that the complexity of the distributed solution approaches the complexity of the centralized one as the number of PZs increases, since there are more scheduling opportunities, and hence less conflicts.

The analysis of the low complexity distributed algorithm follows the same lines as the optimal one, except that the total number of conflicts experienced by \emph{all clouds} is bounded by $U$ according to \corref{cor1}. Therefore, assuming each cloud experience $\lfloor U/C \rfloor$ conflicts, the complexity \emph{per cloud} can be experienced as follows:
\begin{align}
\mathds{C}_{\text{dis}}^{\text{heu}} &= \sum_{|\mathcal{U}_c| = \lfloor U/C \rfloor}^U \alpha^{CBZ|\mathcal{U}_c|} = \alpha^{CBZU}\cfrac{\alpha - \alpha^{BZU(1-C)}}{1-\alpha^{CBZ}} \nonumber \\
&= \mathds{C}_{\text{cen}}^{\text{opt}} \cfrac{\alpha - \alpha^{BZU(1-C)}}{1-\alpha^{CBZ}} \label{eq:com2}
\end{align}
where $1 < \alpha \leq 2$ is a constant that depends on the utilized algorithm. The relative gain in complexity between the optimal and heuristic solutions is given by the following expression:
\begin{align}
\cfrac{\mathds{C}_{\text{dis}}^{\text{heu}}}{\mathds{C}_{\text{dis}}^{\text{opt}}} &= \cfrac{\alpha - \alpha^{BZU(1-C)}}{\alpha^{-CB^2Z} - \alpha^{CBZ}} \ \underset{C \rightarrow \infty}{\longrightarrow} 0 \label{eq:com3}
\end{align}

The above limit shows how the complexity of the heuristic distributed algorithm is negligible as compared to the optimal distributed solution for a large number of clouds. Such complexity simplification comes, however, at the expense of a degradation in the performance as the number of clouds increases, as the simulations section suggests later.

\section{Extremes in Coordination Schemes}\label{sec:ful}

The two extremes in coordination schemes are presented in this section. The fully coordinated system, also known as the signal-level coordinated system, requires a substantial amount of backhaul communication to share all the data streams between the BSs. On the other hand, scheduling-level coordination requires low capacity links to connect all BSs to clouds, as clouds become responsible for determining the scheduling policy of the network only. Although more practical to implement from backhaul requirements perspective, scheduling-level coordination comes at the expense of performance degradation. This section considers the two scheduling policy extremes, i.e., either scheduling-level or signal-level coordination problems. {These two allocation problems are separately considered in the literature, e.g., \cite{6799231,6588350,6786060,6525475,6811617,117665}. This part next shows that the proposed graph theoretical framework developed earlier in this paper can be alternatively used to globally solve the problems. In other words, the scheduling problem in each case can be solved using similar techniques to the one used in solving the original hybrid scheduling problem.}

\subsection{Signal-Level Coordination}

For signal-level coordinated systems, all the data streams of users are shared among the BSs across the network. Hence, a user can be scheduled to many BSs in different clouds. The scheduling problem becomes the one of assigning users to clouds and scheduling them to PZs in each BS frame under the following practical constraints.
\begin{itemize}
\item Each PZ should be allocated to exactly one user.
\item Each user cannot be served by the same PZ across different BSs.
\end{itemize}

Following an analysis similar to the one in \sref{sec:mul}, the scheduling problem can be formulated as a 0-1 mixed integer programming as follows:
\begin{subequations}
\label{full_problem}
\begin{align}
\max& \sum_{c,u,b,z}\pi_{cubz}X_{cubz} \\
\label{constraint6} {\rm s.t.\ }& \sum_{u}X_{cubz}=1, \quad\forall \ (c,b,z)\in \mathcal{C} \times \mathcal{B} \times \mathcal{Z},\\
\label{constraint7}& \sum_{cb} X_{cubz} \leq 1,\quad \forall \ (u,z) \in \mathcal{U} \times \mathcal{Z}, \\
\label{constraint8}& X_{cubz} \in \{0,1\}, \forall \ (c,u,b,z) \in \mathcal{C} \times \mathcal{U}\times \mathcal{B}\times\mathcal{Z},
\end{align}
\end{subequations}
where the optimization is over the binary variable $X_{cubz}$, and where equations \eref{constraint6} and \eref{constraint7} correspond to the first and second system constraints, respectively.

Construct a graph similar to the one constructed in \sref{sec:mul}, except using the connectivity constraints CC2 and CC3 only. Such graph, denoted by $\mathcal{G'}(\mathcal{V'},\mathcal{E'})$, is called here the \textit{reduced conflict graph}. The following lemma provides the optimal solution to the optimization problem \eref{full_problem}.

\begin{lemma}
The optimal solution to the scheduling problem in signal-level coordinated cloud-enabled network \eref{full_problem} is the maximum-weight independent set of size $CBZ$ in the reduced conflict graph which is constructed in a similar manner as the conflict graph but using only connectivity constraint CC2 and CC3.
\label{l1}
\end{lemma}

\begin{proof}
A sketch of the proof goes as follows. The constraints \eref{constraint6}, \eref{constraint7} and \eref{constraint8} of the optimization problem \eref{full_problem} are similar to constraints \eref{constraint3}, \eref{constraint4} and \eref{constraint5}, respectively. Therefore, this lemma can be proved using similar steps of \thref{th1}, except by considering the \textit{reduced conflict graph} $\mathcal{G'}(\mathcal{V'},\mathcal{E'})$ only. A complete proof can be found in \appref{app4}.
\end{proof}

\subsection{Scheduling-Level Coordination}

In scheduling-level coordinated CRAN, the cloud is only responsible for scheduling users to BSs and PZs and synchronizing the transmit frames across the various BSs. In such coordinated systems, the scheduling problem is the one of assigning users to BSs and PZs under the following system constraints:
\begin{itemize}
\item Each user can connect at most to one BS but possibly to many PZs in that BS.
\item Each PZ should be allocated to exactly one user.
\end{itemize}

The scheduling problem can, then, be formulated as follows:
\begin{subequations}
\label{no_problem}
\begin{align}
\max & \sum_{c,u,b,z}\pi_{cubz}X_{cubz} \\
\label{constraint12}{\rm s.t.\ } &Y_{cub} = \min\bigg(\sum_{z}X_{cubz},1\bigg), \quad \forall \ (c,u,b) ,\\
\label{constraint11}& \sum_{c,b}Y_{cub}\leq 1,\quad \forall \ u\in\mathcal{U},\\
\label{constraint13}& \sum_{u}X_{cubz}=1, \quad \forall \ (c,b,z)\in \mathcal{C} \times \mathcal{B}\times\mathcal{Z},\\
\label{constraint14}& X_{cubz},Y_{cub} \in \{0,1\}, \quad \forall \ (c,u,b,z),
\end{align}
\end{subequations}
where the optimization is over the binary variables $X_{cubz}$ and $Y_{cub}$, where the constraints in \eref{constraint12} and \eref{constraint11} correspond to first system constraint, and where the equality constraint in \eref{constraint13} corresponds to the second system constraint.

Construct the scheduling conflict graph $\mathcal{G''}(\mathcal{V''},\mathcal{E''})$ by generating a vertex $v \in \mathcal{V''}$ for each association $a \in \mathcal{A}$. Vertices $v$ and $v^\prime$ are conflicting vertices, and thus connected by an edge in $\mathcal{E''}$ if one of the following connectivity conditions is true:
\begin{itemize}
\item $\delta(\varphi_u(v) - \varphi_u(v^\prime))(1-\delta(\varphi_c(v) - \varphi_c(v^\prime))) = 1$.
\item $(\varphi_c(v),\varphi_b(v),\varphi_z(v)) = (\varphi_c(v^\prime),\varphi_b(v^\prime),\varphi_z(v^\prime))$.
\item $\delta(\varphi_u(v) - \varphi_u(v^\prime))(1-\delta(\varphi_b(v) - \varphi_b(v^\prime))) = 1$.
\end{itemize}

The following proposition characterizes the solution of the scheduling problem in scheduling-level coordinated CRANs:
\begin{proposition}
The optimal solution to the optimization problem \eref{no_problem} is the maximum-weight independent set of size $CBZ$ in the scheduling conflict graph.
\end{proposition}

The proof of this result is omitted as it mirrors the steps used in proving \thref{th1}.

\section{Simulation Results}\label{sec:sim}

\begin{table}
\centering
\caption{System model parameters}
\begin{tabular}{|c|c|}
\hline
Cellular Layout & Hexagonal \\
\hline
Cell-to-Cell Distance & 500 meters \\
\hline
{Channel Model} & SUI-3 Terrain type B \\
\hline
Channel Estimation & Perfect \\
\hline
High Power & -42.60 dBm/Hz \\
\hline
Background Noise Power & -168.60 dBm/Hz \\
\hline
SINR Gap $\Gamma$ & 0dB\\
\hline
Bandwidth & 10 MHz \\
\hline
\end{tabular}
\label{t1}
\end{table}

The performance of the proposed scheduling schemes is shown in this section in the downlink of a cloud-radio access network, similar to \fref{fig:network2}. For illustration purposes, the simulations focus on the sum-rate maximization problem, i.e., $\pi_{cubz} = \log_2 (1 + \text{SINR}^u_{cbz})$. In these simulations, the cell size is set to $500$ meters and users are uniformly placed within each cell. The number of clouds, users, base-stations per cloud and power-zone per base-station frame change in each figure in order to quantify the gain in various scenarios. Simulations parameters are displayed in Table \ref{t1}. {It is crucial to highlight that both the centralized and the distributed optimal algorithms described in Section III and Section IV.A provide the exact same solution denoted by ``Hybrid-level coordination".}

\begin{figure}[t]
\centering
\includegraphics[width=0.9\linewidth]{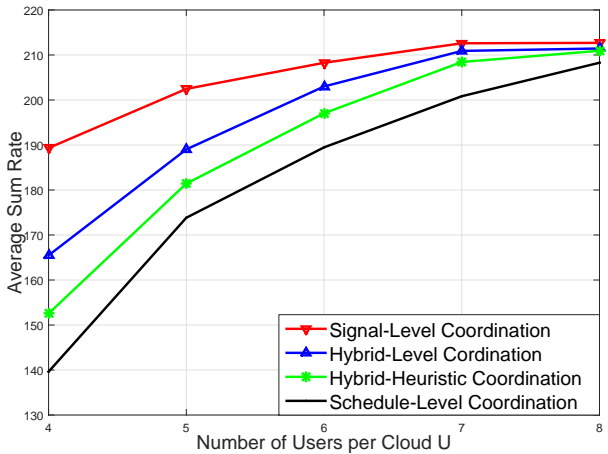}\\
\caption{Sum-rate in bps/Hz versus number of users $U$. Number of clouds is $C=3$ with $B=3$ base-stations per cloud, and $Z=5$ power-zones per BS's transmit frame.}\label{fig:U}
\end{figure}

\begin{figure}[t]
\centering
\includegraphics[width=0.9\linewidth]{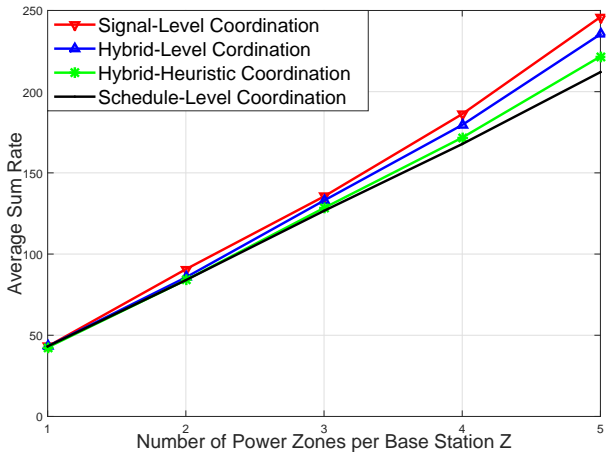}\\
\caption{Sum-rate in bps/Hz versus number of power-zones $Z$ per BS. Number of clouds is $C=3$ with $B=3$ base-stations per cloud, and $U=24$ users.}\label{fig:Z}
\end{figure}

\begin{figure}[t]
\centering
\includegraphics[width=0.9\linewidth]{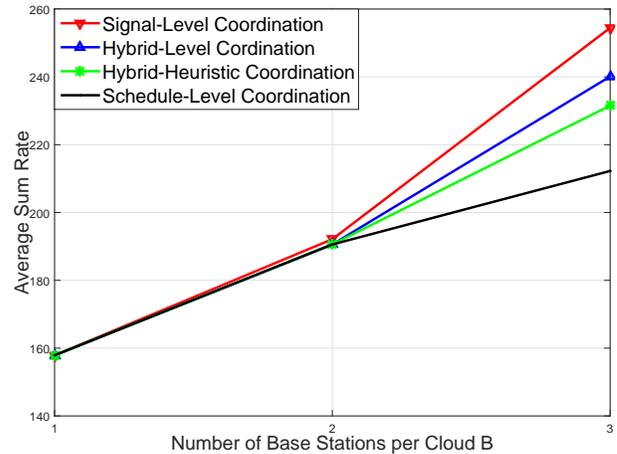}\\
\caption{Sum-rate in bps/Hz versus number of base-stations $B$ per cloud. Number of clouds is $C=3$ with $Z=5$ power-zones per BS's transmit frame, and $U=24$ users.}\label{fig:B}
\end{figure}

\begin{figure}[t]
\centering
\includegraphics[width=0.9\linewidth]{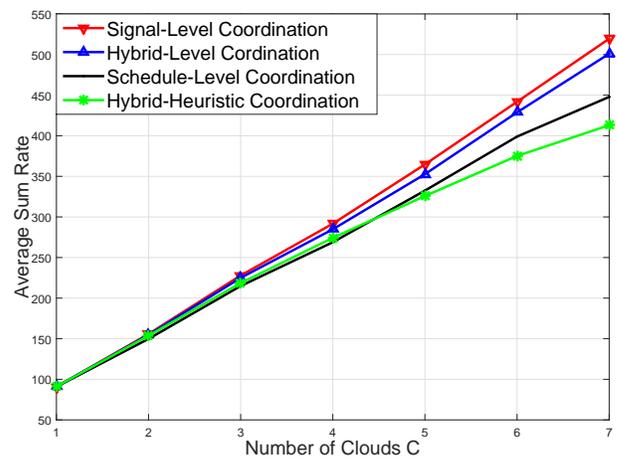}\\
\caption{Sum-rate in bps/Hz versus number of clouds $C$. Number of base-stations is $B=3$ per cloud, with $5$ power-zones per BS's transmit frame, and $U=8$ users per cloud.}\label{fig:C}
\end{figure}

\fref{fig:U} plots the sum-rate in bps/Hz versus the number of users $U$ for a CRAN composed of $C=3$ clouds, $B=3$ base-stations per cloud, and $Z=5$ power-zones per BS's transmit frame. The proposed hybrid coordination policy provides a significant gain against the scheduling-level coordinated system for a small number of users. As the number of users increases in the system, the different strategies performs the same. This can be explained by the fact that as the number of users in the network increases, the probability that different users have the maximum pay-off in various PZs across the network increases, which results in scheduling different users in different PZs and thus the different scheduling policies provide similar performance. The performance of the distributed heuristic approaches the one of the optimal scheduling as the number of users increases. This can be explained by the fact that for a large number of users, the probability that a user is scheduled to more than one cloud decreases, which decreases the conflict among clouds and the likelihood of scheduling user to the wrong cloud.

\fref{fig:Z} plots the sum-rate in bps/Hz versus the number of power-zones $Z$ per BS for a network comprising $C=3$ clouds, $B=3$ base-stations, and $U=24$ users. From the system connectivity of the different policies, we clearly see that for a network comprising only one PZ per BS, the three scheduling policies are equivalent, which explains the similar performance for $Z=1$. As the number of PZs per BS increases, the gap between the different coordinated systems increases. In fact, as the number of PZs increases, the ratio of users per PZ decreases and thus the role of the cloud as a scheduling entity becomes more pronounced.

\fref{fig:B} plots sum-rate in bps/Hz versus the number of base-stations $B$ per cloud for a network comprising $C=3$ clouds, $Z=5$ power-zones per BS's transmit frame, and $U=24$ users. For a small number of BSs, all the policies are equivalent and provide the same gain. However, as this number increases, the higher the level of coordination is, the more scheduling opportunities it offers. This explains the difference in performance as $B$ increases. We can see that our hybrid coordination provides a gain up to $13\%$ as compared to the scheduling-level coordinated network, for a degradation up to $6\%$ as compared to the signal-level coordination.

Finally, \fref{fig:C} plots the sum-rate in bps/Hz versus the number of clouds $C$ for a network comprising $B=3$ base-stations per cloud, $Z=5$ power-zones per BS's transmit frame, and $U=8$ users \emph{per cloud}. Again, our hybrid coordination provides a gain up to $12\%$ as compared with the scheduling-level coordination, for a negligible degradation up to $4\%$ against the signal-level coordinated system. For a large number of clouds, the performance of the distributed heuristic degrades. This can be explained by the fact that for a large number of clouds, the probability that multiple clouds are in conflict for the same user increases, which increases the probability of scheduling users to the wrong cloud; thereby resulting in a performance degradation.

\section{Conclusions}\label{sec:con}

This paper considers the hybrid scheduling problem in the downlink of a multi-cloud radio-access network. The paper maximizes a network-wide utility under the practical constraint that each user is scheduled, at most, to a single cloud, but possibly to many BSs within the cloud and can be served by one or more distinct PZs within the BSs frame. The paper proposes a graph theoretical approach to solving the problem by introducing the conflict graph in which each vertex represents an association of cloud, user, BS and PZ. The problem is then reformulated as a maximum-weight independent set problem that can be efficiently solved. The paper further proposes distributed optimal and heuristic solutions to the coordinated scheduling problem. Finally, the paper shows that the optimal solution to the scheduling problem in different levels of system coordination can be obtained as a special case of the more general proposed system.  Simulation results suggest that the proposed system architecture provides appreciable gain as compared to the scheduling-level coordinated networks, for a negligible degradation against the signal-level coordination.

\appendices
\numberwithin{equation}{section}

\section{Proof of \thref{th1}}\label{app1}

To prove the result, the optimization problem \eref{Original_optimization_problem} is first reformulated as a search over the set of feasible schedules. Further, a one to one mapping between the possible schedules and the set of independent sets of size $Z_{tot}$ in the conflict graph is highlighted. Showing that the weight of each independent set is the objective function of \eref{Original_optimization_problem} indicates that the optimal solution is the maximum-weight independent set which concludes the proof.

All possible schedules representing the assignments between clouds, users, BSs and PZs, regardless of the feasibility, can be conveniently represented by the set of all subsets of $\mathcal{A}$, i.e., the power set $\mathcal{P}(\mathcal{A})$ of the set of associations $\mathcal{A}$. Recall that for an association $a=(c,u,b,z)$ in a schedule $\mathcal{S} \subseteq \mathcal{A}$ (i.e., $\mathcal{S} \in \mathcal{P}(\mathcal{A})$), the benefit of the association is given by $\pi(a)=\pi_{cubz}$. The following lemma reformulates the multi-cloud joint scheduling problem.
\begin{lemma}
The discrete optimization problem \eref{Original_optimization_problem} can be written as follows:
\begin{align}
\label{new_optimization_problem}
&\max_{\mathcal{S} \in \mathcal{P}(\mathcal{A})} \sum_{a \in \mathcal{S}} \pi(a) \\
&{\rm s.t.\ } \mathcal{S} \in \mathcal{F},
\end{align}
where $\mathcal{F}$ is the set of feasible schedules defined as follows:
\begin{subequations}
\begin{align}
&\mathcal{F}= \{ \mathcal{S} \in \mathcal{P}(\mathcal{A}) \text{ such that } \forall \ a \neq a^\prime \in \nonumber \mathcal{S}\\
& \delta(\varphi_u(a) - \varphi_u(a^\prime)) (1-(\delta(\varphi_c(a)-\varphi_c(a^\prime))) = 0, \label{eq1} \\
& (\varphi_c(a),\varphi_b(a),\varphi_z(a)) \neq (\varphi_c(a^\prime),\varphi_b(a^\prime),\varphi_z(a^\prime)) \label{eq2}, \\
& \delta(\varphi_u(a) - \varphi_u(a^\prime))\delta(\varphi_z(a) - \varphi_z(a^\prime)) =0 \label{eq3} \\
&|\mathcal{S}| = Z_{\text{tot}} \label{eq4} \}.
\end{align}
\end{subequations}
\label{l2}
\end{lemma}

\begin{proof}
The proof can be found in \appref{app5}.
\end{proof}

To demonstrate that there is a one to one mapping between the set of feasible schedules $\mathcal{F}$ and the set of independent sets $\mathcal{I}$ of size $Z_{\text{tot}}$, we first show that each element of $\mathcal{F}$ is represented by a unique element in $\mathcal{I}$. We, then, show that each independent set can uniquely be represented by a feasible schedule.

Let the feasible schedule $\mathcal{S} \in \mathcal{F}$ be associated with the set of vertices $I$ in the conflict graph. Assume $\exists \ v \neq v^\prime \in I$ such that $v$ and $v^\prime$ are connected. From the connectivity conditions in the conflict graph, vertices $v$ and $v^\prime$ verify one of the following conditions
\begin{itemize}
\item CC1: $\delta(\varphi_u(v) - \varphi_u(v^\prime))(1-\delta(\varphi_c(v) - \varphi_c(v^\prime))) = 1$: this condition violates the constraint \eref{eq1} of the construction of $\mathcal{F}$.
\item CC2: $(\varphi_c(v),\varphi_b(v),\varphi_z(v)) = (\varphi_c(v^\prime),\varphi_b(v^\prime),\varphi_z(v^\prime))$: this condition violates the constraint \eref{eq2} of the construction of $\mathcal{F}$.
\item CC3: $\delta(\varphi_u(v) - \varphi_u(v^\prime))\delta(\varphi_z(v) - \varphi_z(v^\prime)) = 1$: this condition violates the constraint \eref{eq3} of the construction of $\mathcal{F}$.
\end{itemize}
Therefore, each pair of vertices $v \neq v^\prime \in I$ are not connected which demonstrates that $I$ is an independent set of vertices in the conflict graph. Finally, from the construction constraint \eref{eq4}, $\mathcal{S}$ and by extension $I$ have $Z_{\text{tot}}$ associations. Therefore, $I$ is a set of $Z_{\text{tot}}$ independent vertices which concludes that $I \in \mathcal{I}$. The uniqueness of $I$ follows directly from the bijection between the set of vertices in the graph and the set of associations in $\mathcal{A}$.

To establish the converse, let $I \in \mathcal{I}$ be an independent set of size $Z_{\text{tot}}$ and let $\mathcal{S}$ be its corresponding schedule. Using an argument similar to the one in previous paragraph, it can be easily shown that all the associations in $\mathcal{S}$ verify the constraints \eref{eq1}, \eref{eq2}, and \eref{eq3}. Given that $I$ is of size $Z_{\text{tot}}$, then $\mathcal{S}$ verify \eref{eq4} which concludes that $\mathcal{S} \in \mathcal{F}$. Uniqueness of the element is given by the same argument as earlier.

To conclude the proof, note that the weight of an independent set $I \in \mathcal{I}$ is equal to the objective function \eref{new_optimization_problem} and by extension to the original objective function \eref{Original_optimization_problem}. Therefore, the globally optimal solution of the joint scheduling problem in multi-cloud network \eref{Original_optimization_problem} is equivalent to a maximum-weight independent set among the independent sets of size $Z_{\text{tot}}$ in the conflict graph.

\section{Proof of \thref{th2}}\label{app2}

To show that the distributed solution reached by \algref{algo1} is the optimal solution to the scheduling problem \eref{Original_optimization_problem}, we first show that solving the maximum-weight independent set locally when users are assigned to a cloud like in the optimal solution to \eref{Original_optimization_problem} will yield the optimal solution. Afterward, we show that \algref{algo1} assigns users to clouds as in the optimal solution. Combining the two above points concludes that the solution reached by \algref{algo1} is the optimal solution to \eref{Original_optimization_problem}. To finish the proof, we show that the running time of the algorithm is bounded.

First define $\mathcal{I}_c$ as the set of independent sets of size $BZ$ in the local conflict graph of cloud $c$. Let $I_c \in \mathcal{I}_c$ be an independent sets. The following lemma states the feasibility of the schedule $\mathcal{S} = \bigcup_{c \in \mathcal{C}} I_c$.

\begin{lemma}
Let $I_c \in \mathcal{I}_c$ be an independent set of size $BZ$ in the local conflict graph of cloud $c \in \mathcal{C}$ (i.e., $\varphi_c(v)=c, \ \forall \ v \in I_c$) such that each user is assigned to at most a single cloud. In other words, for $c \neq c^\prime$, we have $\varphi_u(v) \neq \varphi_u(v^\prime), \ \forall \ v \in I_c, \ v^\prime \in I_{c^\prime}$. The schedule $\mathcal{S} = \bigcup_{c \in \mathcal{C}} I_c$ is a feasible solution to the optimization problem \eref{Original_optimization_problem}.
\label{l3}
\end{lemma}

\begin{proof}
The proof can be found in \appref{app6}.
\end{proof}

Let $\mathcal{U}_c^s$ be the set of users scheduled in cloud $c$ at the optimal solution $X_{cubz}^*$ of the optimization problem \eref{Original_optimization_problem}. The mathematical definition of this set is the following:
\begin{align}
\mathcal{U}_c^s = \left\{ u \in \mathcal{U} \ | \ \exists \ (b,z) \in \mathcal{B} \times \mathcal{Z} \text{ such that } X^*_{cubz} = 1\right\}
\end{align}

Showing that solving the maximum-weight independent set locally when users are assigned to cloud like in the optimal solution to \eref{Original_optimization_problem} yields the optimal solution is equivalent to showing the following. Assume that the set allowed users $\mathcal{U}_c^s \subseteq \mathcal{U}_c$ by cloud $c$ is set of users scheduled in that cloud at the optimal solution and let $I_c \in \mathcal{I}_c$ be the maximum-weight independent set of cloud $c$. We have to show that $\mathcal{S} = \bigcup_{c \in \mathcal{C}} I_c$ is the optimal scheduling. According to \thref{th1}, the optimal solution of \eref{Original_optimization_problem} can be written as follows:
\begin{align}
\max_{I \in \mathcal{I}} \sum_{v \in I} w(v)
\end{align}

Let $\mathcal{U}_s = \bigcup_{c \in \mathcal{C}} \mathcal{U}_c^s$ be the set of all scheduled users in the optimal solution. We show that the optimal solution to \eref{Original_optimization_problem} is the same if we consider $\mathcal{U} = \mathcal{U}_s$. Let $\mathcal{I}_s$ be the set of independent sets of size $CBZ$ in the conflict graph $\mathcal{G}(\mathcal{U}_s)$. Therefore, since $\mathcal{I}_s \subseteq \mathcal{I}$, the optimal solution can be written as:
\begin{align}
\max_{I \in \mathcal{I}_s} \sum_{v \in I} w(v) \leq \max_{I \in \mathcal{I}} \sum_{v \in I} w(v)
\label{eq:th26}
\end{align}

However, for the solution $X_{cubz}^*$ we have $\max\limits_{I \in \mathcal{I}_s} \sum_{v \in I} w(v) = \max\limits_{I \in \mathcal{I}} \sum_{v \in I} w(v)$. Therefore, the optimal schedule when considering $\mathcal{U} = \mathcal{U}_s$ is the same as the optimal one of problem \eref{Original_optimization_problem}. We also have $\mathcal{I}_s = \bigcup_{c \in \mathcal{C}} \mathcal{I}_c$. Therefore, the optimal solution of \eref{Original_optimization_problem} can be bounded by the following quantity:
\begin{align}
\max_{I \in \mathcal{I}} \sum_{v \in I} w(v) &= \max_{I \in \mathcal{I}_s} \sum_{v \in I} w(v) \nonumber \\
&= \max_{I \in \bigcup_{c \in \mathcal{C} \mathcal{I}_c}} \sum_{v \in I} w(v) \nonumber \\
& \leq \sum_{c \in \mathcal{C}} \max_{I_c \in \mathcal{I}_c} \sum_{v \in I_c} w(v)
\end{align}

From the feasibility of the optimal solution $X_{cubz}^*$, we have $\mathcal{U}_c^s \cap \mathcal{U}_{c^\prime}^s = \varnothing, \ \forall \ c \neq c^\prime$. Therefore from \lref{l3}, the schedule $\mathcal{S}$ is a feasible solution. In other words, the upper bound is achievable. Therefore, the schedule $\mathcal{S} = \bigcup_{c \in \mathcal{C}} I_c$ is the optimal solution to \eref{Original_optimization_problem}.

We now show that \algref{algo1} assigns users to clouds as in the optimal solution. Assume that a user $u$ scheduled in the cloud $c^*$ in the optimal solution is assigned to that cloud in \algref{algo1}. In other words, we have $u \in \mathcal{U}_{c^*}^s$ and $u \notin \mathcal{U}_{c^*}$. This can happen only if at some time round $t$ in the algorithm, user $u$ is assigned to another cloud $c$. Hence, at some time round we have $u \in \mathcal{K}$, $c^*,c \in \hat{\mathcal{C}}(u)$ and the following equation holds:
\begin{align}
\pi_{cu} + \sum_{\substack{ c^\prime \in \hat{\mathcal{C}}(u) \\ c^\prime \neq c}} \overline{\pi}_{c^\prime u} &\geq \pi_{c^* u} + \sum_{\substack{ c^\prime \in \hat{\mathcal{C}}(u) \\ c^\prime \neq c^*}} \overline{\pi}_{c^\prime u} \nonumber \\
\pi_{cu} + \overline{\pi}_{c^* u} + \sum_{\substack{ c^\prime \in \hat{\mathcal{C}}(u) \\ c^\prime \neq c,c^*}} \overline{\pi}_{c^\prime u} &\geq \pi_{c^* u} + \overline{\pi}_{c u} + \sum_{\substack{ c^\prime \in \hat{\mathcal{C}}(u) \\ c^\prime \neq c,c^*}} \overline{\pi}_{c^\prime u} \nonumber \\
\pi_{cu} + \overline{\pi}_{c^* u} &\geq \pi_{c^* u} + \overline{\pi}_{c u}
\label{eq:th21}
\end{align}

Let $\pi^*$, the objective function of the optimization problem \eref{Original_optimization_problem} at the optimal solution $X_{cubz}^*$, be decomposed as follows:
\begin{align}
\pi^* = \sum_{c^\prime \in \mathcal{C}} \pi^*_{c^\prime} = \sum_{\substack{ c^\prime \in \mathcal{C} \\ c^\prime \neq c,c^*}} \pi^*_{c^\prime} + \pi^*_c + \pi^*_{c^*}
\end{align}

Since user $u$ is scheduled to the cloud $c^*$ in the optimal solution then $\pi^*_{c^*} = \pi^*_{c^*u}$. Moreover, it is clear that $\pi^*_c \leq \overline{\pi}_{cu}$ since $\overline{\pi}_{cu}$ is the optimal schedule for cloud $c$ when it is not allowed to schedule user $u$. Therefore, the optimal objective function of the problem \eref{Original_optimization_problem} is bounded by the following quantity:
\begin{align}
\pi^* \leq \sum_{\substack{ c^\prime \in \mathcal{C} \\ c^\prime \neq c,c^*}} \pi^*_{c^\prime} + \overline{\pi}_{cu} + \pi^*_{c^*u}
\label{eq:th22}
\end{align}

Moreover, it is clear that the merit $\pi_{c^*u}$ of user $u$ that is scheduled to cloud $c^*$, regardless of the feasibility of the whole schedule is higher than any other scheduling feasibility of the entire schedule. In particular, we have:
\begin{align}
\pi^*_{c^*u} \leq \pi_{c^*u}.
\label{eq:th23}
\end{align}
Substituting \eref{eq:th23} in \eref{eq:th22} then applying \eref{eq:th21} yields the following inequality:
\begin{align}
\pi^* &\leq \sum_{\substack{ c^\prime \in \mathcal{C} \\ c^\prime \neq c,c^*}} \pi^*_{c^\prime} + \overline{\pi}_{cu} + \pi_{c^*u} \nonumber \\
&\leq \sum_{\substack{ c^\prime \in \mathcal{C} \\ c^\prime \neq c,c^*}} \pi^*_{c^\prime} + \pi_{cu} + \overline{\pi}_{c^* u}
\label{eq:th24}
\end{align}

Now consider the scheduling in which user $u$ is scheduled to cloud $c$ and all the scheduling for clouds $c^\prime \neq c,c^*$ is the same. The merit function $\pi$ of such scheduling is:
\begin{align}
\pi = \sum_{\substack{ c^\prime \in \mathcal{C} \\ c^\prime \neq c,c^*}} \pi^*_{c^\prime} + \pi_{c} + \pi_{c^*}
\end{align}
Since user $u$ is scheduled to cloud $c$ then $\pi_{c} = \pi_{cu}$.
\begin{align}
\pi = \sum_{\substack{ c^\prime \in \mathcal{C} \\ c^\prime \neq c,c^*}} \pi^*_{c^\prime} + \pi_{cu} + \pi_{c^*}
\end{align}
The merit $\overline{\pi}_{c^* u}$ being the optimal benefit when user $c^*$ is not allowed to schedule user $u$ that it is greater than the merit of any schedule that do not schedule user $u$. In particular, since in $\pi_{c^*}$ user $u$ is not scheduled to cloud $c^*$ then we obtain:
\begin{align}
\overline{\pi}_{c^* u} \leq \pi_{c^*}
\label{eq:th25}
\end{align}
Substituting \eref{eq:th25} in \eref{eq:th24}, we obtain:
\begin{align}
\pi^* &\leq \sum_{\substack{ c^\prime \in \mathcal{C} \\ c^\prime \neq c,c^*}} \pi^*_{c^\prime} + \pi_{cu} + \pi_{c^*} \leq \pi,
\end{align}
which is in contradiction with the fact that $pi^*$ is the optimal weight that that $\pi$ is the merit of a feasible schedule. Finally, we conclude that \eref{eq:th21} do not hold, and that \algref{algo1} assigns users to clouds as in the optimal solution.

To show that the optimal solution is reached by \algref{algo1} we combine the previous two results. First note that when the algorithm terminates we have $\mathcal{K} = \varnothing$. Let $I_c$ the maximum-weight clique in each cloud. Using a proof similar to the one in \eref{eq:th26}, we can easily show that the maximum-weight clique $I_c$ do not change if we consider the set $\tilde{\mathcal{U}}_c$ of users used in the scheduling $I_c$ instead of $\mathcal{U}_c$. Since $\mathcal{K} = \varnothing$ then $\tilde{\mathcal{U}}_c \cap \tilde{\mathcal{U}}_{c^\prime} = \varnothing$. Moreover, we show above that $\mathcal{U}_c^s \subseteq \tilde{\mathcal{U}}_c$. As shown earlier, this condition is equivalent to solving optimally the scheduling problem. Finally, the optimal solution can be reached by \algref{algo1}.

To show that the running time of the algorithm is bounded it this sufficient to note that at each time round of the algorithm, since $\mathcal{K} \neq \varnothing$, then $\exists \ c,u$ such that $\mathcal{U}_c = \mathcal{U}_c \setminus \{u\}$. In other words, $|\mathcal{U}_c| = |\mathcal{U}_c|-1$ Since that $|\mathcal{U}_c|$ is lower bounded by $|\mathcal{U}_c^s|$. Therefore, the running time of \algref{algo1} is bounded by $C (\max_c |\mathcal{U}_c| - \min_c |\mathcal{U}_c^s|)$. Clearly, we have $\max_c |\mathcal{U}_c| = U$. We now show that $\min_c |\mathcal{U}_c^s|=Z$.

We show that for a schedule $\mathcal{S}$ to be feasible, a user $u$ assigned to cloud $c$ can be scheduled to at most $Z$ PZs across the different BSs in $c$. Assume that user $u$ is connected to $Z^\prime > Z$ PZs then the schedule $\mathcal{S}$ contains $Z^\prime$ vertices $v$ such that $\varphi_u(v) =u$. The number of PZ index being $Z$ then from the pigeon-hole principle $\exists \ v,v^\prime$ such that $\varphi_z(v) =\varphi_z(v^\prime)$. From the graph connectivity condition C3, we have $\delta(\varphi_u(v) - \varphi_u(v^\prime))\delta(\varphi_z(v) - \varphi_z(v^\prime)) = 1$. Therefore, vertices $v$ and $v^\prime$ are connected which is in contradiction with the fact that the schedule $\mathcal{S}$ is a feasible solution and hence an independent set. Finally, The running time of the algorithm is bounded by $C(U-B)$. Note that since $U \geq B$ for the problem to have at least one solution, then the quantity is always positive.

\section{Proof of \corref{cor1}}\label{app3}

To prove this corollary, it is sufficient to show that \algref{algo2} converges. Afterwards, applying the result of \lref{l3} guarantee the feasibility of the solution. At time round of the algorithm such that $\mathcal{K} \neq \varnothing$, we have $\exists \ c^*,u$ such that $\forall \ c \neq c^*$ we have $\mathcal{U}_c = \mathcal{U}_c \setminus \{u\}$. Therefore, the running time of the algorithm is bounded by $\max_c |\mathcal{U}_c|$ which is equal to $U$ from \thref{th2}. Therefore, \algref{algo2} converges and outputs the independent sets $I_c \in \mathcal{I}_c$. From \lref{l3}, the solution $\mathcal{S} = \bigcup_{c \in \mathcal{C}} I_c$ is a feasible solution to the optimization problem \eref{Original_optimization_problem} since $\mathcal{U}_c \cap \mathcal{U}_{c^\prime} = \varnothing ,\ \forall \ c\neq c^\prime$.

\section{Proof of \lref{l1}}\label{app4}

Note that the constraints \eref{constraint6}, \eref{constraint7} and \eref{constraint8} of the optimization problem \eref{full_problem} are the same constraints as \eref{constraint3}, \eref{constraint4} and \eref{constraint5}, respectively, in the original optimization problem \eref{Original_optimization_problem}. Therefore, this lemma can be proved using steps similar to the one used in \thref{th1}.

Let $\mathcal{F} \subset \mathcal{P}(\mathcal{A})$ be the set of feasible schedules. Given the mapping between the original constraints of the problem and the constraints of constructing the set $\mathcal{F}$ illustrated in \lref{l2}, it can be easily shown that problem \eref{full_problem} can be written as follows:
\begin{align}
&\max_{\mathcal{S} \in \mathcal{P}(\mathcal{A})} \sum_{a \in \mathcal{S}} \pi(a) \\
&{\rm s.t.\ } \mathcal{S} \in \mathcal{F},
\end{align}
where $\mathcal{F}$ is the set of feasible schedules defined as follows:
\begin{subequations}
\begin{align}
&\mathcal{F}= \{ \mathcal{S} \in \mathcal{P}(\mathcal{A}) \text{ such that } \forall \ a \neq a^\prime \in \nonumber \mathcal{S}\\
& (\varphi_c(a),\varphi_b(a),\varphi_z(a)) \neq (\varphi_c(a^\prime),\varphi_b(a^\prime),\varphi_z(a^\prime)) , \\
& \delta(\varphi_u(a) - \varphi_u(a^\prime))\delta(\varphi_z(a) - \varphi_z(a^\prime)) =0 \\
&|\mathcal{S}| = Z_{\text{tot}} \}.
\end{align}
\end{subequations}

Let the reduced conflict graph be constructed by generating a vertex of each association $a \in \mathcal{A}$ and connecting two distinct vertices $v$ and $v^\prime$ if one of the following two conditions holds:
\begin{itemize}
\item CC2: $(\varphi_c(v),\varphi_b(v),\varphi_z(v)) = (\varphi_c(v^\prime),\varphi_b(v^\prime),\varphi_z(v^\prime))$.
\item CC3: $\delta(\varphi_u(v) - \varphi_u(v^\prime))\delta(\varphi_z(v) - \varphi_z(v^\prime)) = 1$.
\end{itemize}

Define $\mathcal{I}$ as the set of the independent set of vertices of size $Z_{\text{tot}}$ in the reduced conflict graph. Following steps similar to the one used in \thref{th1}, it can be shown that there is a one to one mapping between the set of feasible schedule $\mathcal{F}$ and the set $\mathcal{I}$ and that the objective function is represented by the sum of the weight of the vertices in the independent set. As a conclusion, the optimal solution to the scheduling problem \eref{full_problem} in signal-level coordinated cloud-enabled network is the maximum-weight independent set of size $CBZ$ in the reduced conflict graph.

\section{Proof of \lref{l2}}\label{app5}

To prove this lemma, it is sufficient to prove to that the objective function and the constraints of \eref{Original_optimization_problem} are equivalent to those of the optimization problem \eref{new_optimization_problem}. The objective function of \eref{Original_optimization_problem} is equivalent to the one of \eref{new_optimization_problem} as shown in the following equation:
\begin{align}
\sum_{c,u,b,z}\pi_{cubz}X_{cubz} = \sum_{a \in \mathcal{A}}\pi(a)X(a) = \sum_{a \in \mathcal{S}}\pi(a),
\end{align}
where $X(a)$ is defined in the same manner as $\pi(a)$, i.e., $X(a)= X_{cubz}$ for $a=(c,u,b,z) \in \mathcal{A}$ and $\mathcal{S} =\{a \in \mathcal{A} \ | \ X(a)=1 \}$. Therefore, the two objective functions are equivalent:
\begin{align}
\max \sum_{c,u,b,z}\pi_{cubz}X_{cubz} = \max_{\mathcal{S} \in \mathcal{P}(\mathcal{A})} \sum_{a \in \mathcal{S}}\pi(a).
\end{align}

In what follows, the constraints \eref{constraint1} and \eref{constraint2} are shown to be equivalent to the constraint \eref{eq1}, the constraint \eref{constraint3} is proven to be equivalent to \eref{eq2} and \eref{eq4}. Finally to conclude the proof, \eref{constraint4} is demonstrated to be the same constraint as \eref{eq3}.

Define $\mathcal{S}_{cu} \subset \mathcal{S}$ as the set of associations in schedule $\mathcal{S}$ concerning the $c$th cloud and the $u$th user. The expression of the set is the following:
\begin{align}
\mathcal{S}_{cu} = \left\{ a \in \mathcal{S} \ | \ \varphi_c(a)=c ,\ \varphi_u(a)=u \right\}.
\end{align}
Let $\mathcal{S}_{u} \subset \mathcal{P}(\mathcal{S})$ be the set of all the set concerning user $u$ defined as:
\begin{align}
\mathcal{S}_{u} = \left\{ \mathcal{S}_{cu}, c \in \mathcal{C} \right\}.
\end{align}
The constraints \eref{constraint1} (i.e., $Z_{cu} = 1 - \delta \bigg(\sum_{b,z}X_{cubz}\bigg)$) and \eref{constraint2} (i.e., $\sum_{c}Z_{cu} \leq 1$) are equivalent to the following constraint
\begin{align}
Z_{cu} = 1 - \delta \bigg(\sum_{b,z}X_{cubz}\bigg) \leq 1 \Leftrightarrow |\mathcal{S}_u| \leq 1 , \ \forall \ u.
\label{l2:eq6}
\end{align}
We now show that the inequality $|\mathcal{S}_u| \leq 1$ is equivalent to the following equality $\forall \ a \neq a^\prime \in \mathcal{S}$:
\begin{align}
\delta(\varphi_u(a) - \varphi_u(a^\prime)) (1-(\delta(\varphi_c(a)-\varphi_c(a^\prime))) = 0
\label{l2:eq7}
\end{align}

First note that if $a \in \mathcal{S}_{u}$ and $a^\prime \in \mathcal{S}_{u^\prime}$ with $u \neq u^\prime$, then $\varphi_u(a) \neq \varphi_u(a^\prime)$ which concludes that \eref{l2:eq7} holds for such $a$ and $a^\prime$. Now let $a \neq a^\prime \in \mathcal{S}_{u}$. Since $|\mathcal{S}_u| \leq 1$ then $\exists$ unique $c \in \mathcal{C}$ such that $\mathcal{S}_{cu} \neq \varnothing$. Hence $a \neq a^\prime \in \mathcal{S}_{cu}$, i.e., $\varphi_c(a) = \varphi_c(a^\prime)$ which concludes that \eref{l2:eq7} holds for such $a$ and $a^\prime$. Given that $\mathcal{S}$ can be written as $\bigcup_{u} \mathcal{S}_u$, then \eref{l2:eq7} is valid $\forall \ a \neq a^\prime \in \mathcal{S}$. Combining \eref{l2:eq6} and \eref{l2:eq7} proves that the constraints \eref{constraint1} and \eref{constraint2} are equivalent to the constraint \eref{eq1}.

Define $\mathcal{S}_{cbz} \subset \mathcal{S}$ as the set of associations in schedule $\mathcal{S}$ concerning the $z$th PZ in the $b$th BS connected to the $c$th cloud. The expression of the set is the following:
\begin{align}
\mathcal{S}_{cbz} = \left\{ a \in \mathcal{S} \ | \ \varphi_c(a)=c ,\ \varphi_b(a)=b ,\ \varphi_z(a)=z \right\}.
\end{align}
The constraint \eref{constraint3} can be written as a function of the partial schedules as follows:
\begin{align}
\sum_{u}X_{cubz}=1 \Leftrightarrow |\mathcal{S}_{cbz}|=1, \ \forall \ (c,b,z).
\label{l2:eq1}
\end{align}
Assume $\exists \ a \neq a^\prime \in \mathcal{S}$ such that $\varphi_c(a)=\varphi_c(a^\prime) ,\ \varphi_b(a)=\varphi_b(a^\prime)$ , and $\varphi_z(a)=\varphi_z(a^\prime)$. It is clear that $a,a^\prime \in \mathcal{S}_{cbz}$ where $c=\varphi_c(a) ,\ b=\varphi_b(a)$, and $z=\varphi_z(a)$. However, from \eref{l2:eq1}, we have $|\mathcal{S}_{cbz}|=1$. Therefore, $a = a^\prime$ which concludes that, $\forall \ a \neq a^\prime \in \mathcal{S}$, we have:
\begin{align}
(\varphi_c(a),\varphi_b(a),\varphi_z(a)) \neq (\varphi_c(a^\prime),\varphi_b(a^\prime),\varphi_z(a^\prime)).
\label{l2:eq2}
\end{align}
We now show that $\mathcal{S}_{cbz} \cap \mathcal{S}_{c^\prime b^\prime z^\prime} = \varnothing$ for all sets in which at least one of the following holds: $c \neq c^\prime$, and/or $b \neq b^\prime$, and/or $z \neq z^\prime$. From \eref{l2:eq1}, both sets contain a single association, hence $\mathcal{S}_{cbz} \cap \mathcal{S}_{c^\prime b^\prime z^\prime} \neq \varnothing$ means that $\mathcal{S}_{cbz} = \mathcal{S}_{c^\prime b^\prime z^\prime}$ which do not hold since $c \neq c^\prime$, and/or $b \neq b^\prime$, and/or $z \neq z^\prime$. As a conclusion, the cardinality of the schedule $\mathcal{S}$ can be written as:
\begin{align}
|\mathcal{S}| = \left|\bigcup_{c,b,z} \mathcal{S}_{cbz} \right| = \bigcup_{c,b,z} \left|\mathcal{S}_{cbz} \right| = CBZ = Z_{\text{tot}} .
\label{l2:eq3}
\end{align}
The combination of equations \eref{l2:eq1}, \eref{l2:eq2} and \eref{l2:eq3} shows that the constraint \eref{constraint3} is equivalent to \eref{eq2} and \eref{eq4}.

Define $\mathcal{S}_{uz} \subset \mathcal{S}$ as the set of associations in schedule $\mathcal{S}$ concerning the $u$th user scheduled in the $z$th PZ of one of the connected BS. The expression of the set is the following:
\begin{align}
\mathcal{S}_{uz} = \left\{ a \in \mathcal{S} \ | \ \varphi_u(a)=u ,\ \varphi_z(a)=z \right\}.
\end{align}
The constraint \eref{constraint4} can be written as a function of the partial schedules as follows:
\begin{align}
Y_{uz} = \sum_{cb} X_{cubz} \leq 1 \Leftrightarrow |\mathcal{S}_{uz}| \leq 1, \ \forall \ (u,z).
\label{l2:eq4}
\end{align}
To conclude the proof, it is sufficient to show that, if $ |\mathcal{S}_{uz}| \leq 1, \ \forall \ (u,z)$, then the following equation holds for $\ a \neq a^\prime \in \mathcal{S}$:
\begin{align}
\delta(\varphi_u(a) - \varphi_u(a^\prime))\delta(\varphi_z(a) - \varphi_z(a^\prime)) =0.
\label{l2:eq5}
\end{align}
Let the schedule be partitioned into partial schedules as follows $\mathcal{S} = \bigcup_{uz}\mathcal{S}_{uz}$. For $a \in \mathcal{S}_{uz}$ and $a^\prime \in \mathcal{S}_{u^\prime z^\prime} \neq \mathcal{S}_{uz}$, it is clear that either $u \neq u^\prime$ and/or $z \neq z^\prime$. Hence, equality \eref{l2:eq5} holds for all $a \in \mathcal{S}_{uz}$ and $a^\prime \in \mathcal{S}_{u^\prime z^\prime} \neq \mathcal{S}_{uz}$. Given that $|\mathcal{S}_{uz}| \leq 1$, then $\nexists \ a \neq a^\prime \in \mathcal{S}_{uz}, \ \forall \ (u,z)$ which concludes that \eref{l2:eq5} is verified. The combination of equations \eref{l2:eq4}, and \eref{l2:eq5} shows that the constraint \eref{constraint4} is equivalent to \eref{eq3}.

\section{Proof of \lref{l3}} \label{app6}

To show this lemma, according to \thref{th1}, we only need to show that the schedule $\mathcal{S} = \bigcup_{c \in \mathcal{C}} I_c$ is an independent set of size $CBZ$ in the conflict graph. Since $I_c, \ \forall \ c \in \mathcal{C}$ is an independent set in the local conflict graph then proving that $\mathcal{S}$ is an independent set in the conflict graph boils down to proving that there are no connections between any pair of vertices belonging to different local independent set $I_c$ and $I_{c^\prime}, \ c \neq c^\prime$.

Let $\mathcal{G}_c$ and $\mathcal{G}_{c^\prime}$ be two distinct local conflict graphs (i.e., $c \neq c^\prime$). We show that if $v \in \mathcal{G}_c$ and $v^\prime \in \mathcal{G}_{c^\prime}$ are connected then $\varphi_u(v) = \varphi_u(v^\prime)$. From the connectivity conditions of vertices, $v$ and $v^\prime$ are connected if and only if at least one of the following conditions is verified:
\begin{itemize}
\item C1: $\delta(\varphi_u(v) - \varphi_u(v^\prime))(1-\delta(\varphi_c(v) - \varphi_c(v^\prime))) = 1$.
\item C2: $(\varphi_c(v),\varphi_b(v),\varphi_z(v)) = (\varphi_c(v^\prime),\varphi_b(v^\prime),\varphi_z(v^\prime))$.
\item C3: $\delta(\varphi_u(v) - \varphi_u(v^\prime))\delta(\varphi_z(v) - \varphi_z(v^\prime)) = 1$.
\end{itemize}

Clearly condition C2 cannot be satisfied since $\varphi_c(v)=c \neq c^\prime = \varphi_c(v^\prime)$. Now assume that $\varphi_u(v) \neq \varphi_u(v^\prime)$, then $\delta(\varphi_u(v) - \varphi_u(v^\prime))=0$. This last equality concludes that conditions C1 and C3 are not satisfied and hence the vertices not connected, which is a contradiction with the initial assumption. Therefore, $\varphi_u(v) = \varphi_u(v^\prime)$ for vertices $v$ and $v^\prime$ belonging to different local conflict graphs $\mathcal{G}_c$ and $\mathcal{G}_{c^\prime}$.

Given that in the schedule $\mathcal{S}$, we have $\varphi_u(v) \neq \varphi_u(v^\prime), \ \forall \ v \in I_c, \ v^\prime \in I_{c^\prime}$ then there are no connections between any pair of vertices belonging to different local independent set $I_c$ and $I_{c^\prime}$. Therefore, $\mathcal{S}$ is an independent set in the conflict graph which size is equal to the sum of size of the local independent sets $I_c$. In other words, $\mathcal{S}$ is an independent set of size $Z_{\text{tot}} = CBZ$ which concludes that it is a feasible solution to the optimization problem \eref{Original_optimization_problem}.

\bibliographystyle{IEEEtran}
\bibliography{citations}

\begin{thebibliography}{10}
\providecommand{\url}[1]{#1}
\csname url@samestyle\endcsname
\providecommand{\newblock}{\relax}
\providecommand{\bibinfo}[2]{#2}
\providecommand{\BIBentrySTDinterwordspacing}{\spaceskip=0pt\relax}
\providecommand{\BIBentryALTinterwordstretchfactor}{4}
\providecommand{\BIBentryALTinterwordspacing}{\spaceskip=\fontdimen2\font plus
\BIBentryALTinterwordstretchfactor\fontdimen3\font minus
  \fontdimen4\font\relax}
\providecommand{\BIBforeignlanguage}[2]{{%
\expandafter\ifx\csname l@#1\endcsname\relax
\typeout{** WARNING: IEEEtran.bst: No hyphenation pattern has been}%
\typeout{** loaded for the language `#1'. Using the pattern for}%
\typeout{** the default language instead.}%
\else
\language=\csname l@#1\endcsname
\fi
#2}}
\providecommand{\BIBdecl}{\relax}
\BIBdecl

\bibitem{3514852}
A.~Douik, H.~Dahrouj, T.~Y. Al-Naffouri, and M.-S. Alouini, ``Multi-cloud
  coordinating via joint scheduling for the downlink of radio-access
  networks,'' \emph{IEEE Global Telecommunications Conference (GLOBECOM' 2015),
  San Diego, CA, USA, available online http://arxiv.org/abs/1504.01552}, 2015.

\bibitem{6824752}
J.~Andrews, S.~Buzzi, W.~Choi, S.~Hanly, A.~Lozano, A.~Soong, and J.~Zhang,
  ``What will 5{G} be?'' \emph{IEEE Journal on Selected Areas in
  Communications}, vol.~32, no.~6, pp. 1065--1082, June 2014.

\bibitem{6736746}
F.~Boccardi, R.~Heath, A.~Lozano, T.~Marzetta, and P.~Popovski, ``Five
  disruptive technology directions for 5{G},'' \emph{IEEE Communications
  Magazine}, vol.~52, no.~2, pp. 74--80, February 2014.

\bibitem{6983623}
M.~Peng, K.~Zhang, J.~Jiang, J.~Wang, and W.~Wang, ``Energy-efficient resource
  assignment and power allocation in heterogeneous cloud radio access
  networks,'' \emph{IEEE Transactions on Vehicular Technology}, vol.~PP,
  no.~99, pp. 1--1, 2014.

\bibitem{5594708}
D.~Gesbert, S.~Hanly, H.~Huang, S.~Shamai~Shitz, O.~Simeone, and W.~Yu,
  ``Multi-cell mimo cooperative networks: A new look at interference,''
  \emph{IEEE Journal on Selected Areas in Communications}, vol.~28, no.~9, pp.
  1380--1408, December 2010.

\bibitem{7143328}
H.~Dahrouj, A.~Douik, O.~Dhifallah, T.~Y. Al-Naffouri, and M.-S. Alouini,
  ``Resource allocation in heterogeneous cloud radio access networks: advances
  and challenges,'' \emph{IEEE Wireless Communications}, vol.~22, no.~3, pp.
  66--73, June 2015.

\bibitem{6799231}
S.-H. Park, O.~Simeone, O.~Sahin, and S.~Shamai, ``Inter-cluster design of
  precoding and fronthaul compression for cloud radio access networks,''
  \emph{IEEE Wireless Communications Letters}, vol.~3, no.~4, pp. 369--372, Aug
  2014.

\bibitem{6588350}
------, ``Joint precoding and multivariate backhaul compression for the
  downlink of cloud radio access networks,'' \emph{IEEE Transactions on Signal
  Processing}, vol.~61, no.~22, pp. 5646--5658, Nov 2013.

\bibitem{6786060}
Y.~Shi, J.~Zhang, and K.~Letaief, ``Group sparse beamforming for green
  cloud-ran,'' \emph{IEEE Transactions on Wireless Communications}, vol.~13,
  no.~5, pp. 2809--2823, May 2014.

\bibitem{6525475}
W.~Yu, T.~Kwon, and C.~Shin, ``Multicell coordination via joint scheduling,
  beamforming, and power spectrum adaptation,'' \emph{IEEE Transactions on
  Wireless Communications}, vol.~12, no.~7, pp. 1--14, July 2013.

\bibitem{6811617}
H.~Dahrouj, W.~Yu, T.~Tang, J.~Chow, and R.~Selea, ``Coordinated scheduling for
  wireless backhaul networks with soft frequency reuse,'' in \emph{Proc. of the
  21st Europea Signal Processing Conference (EUSIPCO' 2013), Marrakech,
  Morocco}, Sept 2013, pp. 1--5.

\bibitem{117665}
A.~Douik, H.~Dahrouj, T.~Y. Al-Naffouri, and M.-S. Alouini, ``Coordinated
  scheduling for the downlink of cloud radio-access networks,'' \emph{Proc. of
  {IEEE} International Conference on Communications (ICC' 2015), London, UK.},
  2015.

\bibitem{5464705}
B.~Rengarajan, A.~Stolyar, and H.~Viswanathan, ``Self-organizing dynamic
  fractional frequency reuse on the uplink of ofdma systems,'' in \emph{Proc.
  of 2010 44th Annual Conference on Information Sciences and Systems (CISS'
  2010), Princeton, New Jersey, USA}, March 2010, pp. 1--6.

\bibitem{49842514}
D.~P. Bertsekas, ``The auction algorithm: {A} distributed relaxation method for
  the assignment problem,'' \emph{Annals of Operations Research}, vol.~14, pp.
  105--123, 1988.

\bibitem{ouss_glob1}
O.~Dhifallah, H.~Dahrouj, T.~Y. Al-Naffouri, and M.-S. Alouini, ``Decentralized
  group sparse beamforming for multi-cloud radio access networks,'' in
  \emph{Proc. of IEEE Globecom}, San Diego, USA, Dec 2015.

\bibitem{4450840}
W.~Choi and J.~Andrews, ``The capacity gain from intercell scheduling in
  multi-antenna systems,'' \emph{IEEE Transactions on Wireless Communications},
  vol.~7, no.~2, pp. 714--725, February 2008.

\bibitem{6831362}
B.~Dai and W.~Yu, ``Sparse beamforming for limited-backhaul network mimo system
  via reweighted power minimization,'' in \emph{Proc. of IEEE Global
  Telecommunications Conference (GLOBECOM' 2013), Atlanta, GA, USA}, Dec 2013,
  pp. 1962--1967.

\bibitem{4432271}
S.~Kiani and D.~Gesbert, ``Optimal and distributed scheduling for multicell
  capacity maximization,'' \emph{IEEE Transactions on Wireless Communications},
  vol.~7, no.~1, pp. 288--297, Jan 2008.

\bibitem{5199027}
R.~Bendlin, Y.-F. Huang, M.~Ivrlac, and J.~Nossek, ``Fast distributed
  multi-cell scheduling with delayed limited-capacity backhaul links,'' in
  \emph{Proc. of IEEE International Conference on Communications (ICC' 2009),
  Dresden, Germany}, June 2009, pp. 1--5.

\bibitem{5165179}
J.~Papandriopoulos and J.~Evans, ``Scale: A low-complexity distributed protocol
  for spectrum balancing in multiuser dsl networks,'' \emph{IEEE Transactions
  on Information Theory}, vol.~55, no.~8, pp. 3711--3724, Aug 2009.

\bibitem{4686832}
J.~Mundarath, P.~Ramanathan, and B.~Van~Veen, ``A distributed downlink
  scheduling method for multi-user communication with zero-forcing
  beamforming,'' \emph{IEEE Transactions on Wireless Communications}, vol.~7,
  no.~11, pp. 4508--4521, November 2008.

\bibitem{6086561}
Y.~Xu, J.~Wang, Q.~Wu, A.~Anpalagan, and Y.~D. Yao, ``Opportunistic spectrum
  access in cognitive radio networks: Global optimization using local
  interaction games,'' \emph{IEEE Journal of Selected Topics in Signal
  Processing}, vol.~6, no.~2, pp. 180--194, April 2012.

\bibitem{6492306}
Y.~Xu, A.~Anpalagan, Q.~Wu, L.~Shen, Z.~Gao, and J.~Wang, ``Decision-theoretic
  distributed channel selection for opportunistic spectrum access: Strategies,
  challenges and solutions,'' \emph{IEEE Communications Surveys Tutorials},
  vol.~15, no.~4, pp. 1689--1713, Fourth 2013.

\bibitem{4389757}
D.~Gesbert, S.~Kiani, A.~Gjendemsjo, and G.~ien, ``Adaptation, coordination,
  and distributed resource allocation in interference-limited wireless
  networks,'' \emph{Proceedings of the IEEE}, vol.~95, no.~12, pp. 2393--2409,
  Dec 2007.

\bibitem{1626432}
J.~Huang, R.~Berry, and M.~Honig, ``Distributed interference compensation for
  wireless networks,'' \emph{IEEE Journal on Selected Areas in Communications},
  vol.~24, no.~5, pp. 1074--1084, May 2006.

\bibitem{4558622}
C.~Shi, R.~Berry, and M.~Honig, ``Distributed interference pricing for ofdm
  wireless networks with non-separable utilities,'' in \emph{Proc. of 42nd
  Annual Conference on Information Sciences and Systems (CISS' 2008),
  Princeton, New Jersey, USA}, March 2008, pp. 755--760.

\bibitem{4151582}
J.~Yuan and W.~Yu, ``Distributed cross-layer optimization of wireless sensor
  networks: A game theoretic approach,'' in \emph{Proc. of IEEE Global
  Telecommunications Conference (GLOBECOM' 2006), San Francisco, California,
  USA}, Nov 2006, pp. 1--5.

\bibitem{15522856}
F.~V. Fomin, F.~Grandoni, and D.~Kratsch, ``A measure \& conquer approach for
  the analysis of exact algorithms,'' \emph{Journal of the ACM}, vol.~56,
  no.~5, pp. 25:1--25:32, Aug. 2009.

\bibitem{2155446}
N.~Bourgeois, B.~Escoffier, V.~T. Paschos, and J.~M.~M. van Rooij, ``A
  bottom-up method and fast algorithms for max independent set,'' in
  \emph{Proc. of the 12th Scandinavian Conference on Algorithm Theory (SWAT'
  2010), Bergen, Norway}.

\bibitem{6848102}
P.~jun Wan, X.~Jia, G.~Dai, H.~Du, and O.~Frieder, ``Fast and simple
  approximation algorithms for maximum weighted independent set of links,'' in
  \emph{Proc. of 33th IEEE Conference on Computer Communications (INFOCOM'
  2009), Toronto, canada}, April 2014, pp. 1653--1661.

\bibitem{5341909}
N.~Esfahani, P.~Mazrooei, K.~Mahdaviani, and B.~Omoomi, ``A note on the p-time
  algorithms for solving the maximum independent set problem,'' in \emph{Proc.
  of 2nd Conference on Data Mining and Optimization (DMO' 2009), Bandar Baru
  Bangi, Malaysia}, Oct 2009, pp. 65--70.

\end{thebibliography}

\end{document}